\documentclass[acmsmall, screen]{acmart}
\AtBeginDocument{%
  }

\setcopyright{acmlicensed}
\copyrightyear{2018}
\acmYear{2018}
\acmDOI{XXXXXXX.XXXXXXX}
\acmConference[Conference acronym 'XX]{Make sure to enter the correct
  conference title from your rights confirmation email}{June 03--05,
  2018}{Woodstock, NY}
\acmISBN{978-1-4503-XXXX-X/2018/06}

\usepackage{amsfonts}
\usepackage{amsmath}
\usepackage{hyperref}
\usepackage{thmtools}
\usepackage{mathtools}
\usepackage{cmll}
\usepackage{cleveref}
\usepackage[bbgreekl]{mathbbol}
\usepackage{virginialake}
\vlnostructuressyntax

\usepackage{proofzilla}
\usepackage{xspace}
\usepackage{adjustbox}

\usepackage[backgroundcolor=lightgray]{todonotes}

\newcommand{\scalevlderivation}[2][1]{
	\scalebox{#1}{
		$\vlderivation{#2}$
	}
}

\def\chor{choreography\xspace}
\def\chors{choreographies\xspace}

\def\net{network\xspace}
\def\nets{networks\xspace}

\def\derivation{derivation\xspace}
\def\cyclic{cyclic\xspace}

\def\regular{regular\xspace}

\def\thread{thread\xspace}

\def\progressing{progressing\xspace}

\def\LF{lock-free\xspace}

\def\redtree{reduction tree\xspace}
\def\redtrees{reduction trees\xspace}
\def\fair{fair\xspace}
\def\fairness{fairness\xspace}

\def\chorok{well-formed\xspace}

\def\nice{nice\xspace}
\def\niceness{niceness\xspace}
\def\ok{fairly \LF}

\def\defn#1{\emph{#1}}
\def\resp#1{(resp.~{#1})\xspace}

\def\qquand{\quad\text{and}\quad}

\def\qomma{\;,\;}
\def\qqomma{\quad,\quad}
\def\qwith{\;\text{with}\;}
\def\qqwith{\quad\text{with}\quad}

\def\rclr#1{{\color{red}#1}}
\def\bclr#1{{\color{blue}#1}}

\def\nclr#1{{\color{pzbrickred}#1}}

\def\gclr#1{{\color{pzgreen}#1}}

\mathchardef\mhyphen="2D

\def\set#1{\{#1\}}
\def\Set#1{\left\{#1\right\}}
\def\intset#1#2{\set{#1, \ldots, #2}}

\def\sizeof#1{\left|#1\right|}

\def\fsubst#1#2{\left[{#1}/{#2} \right]}
\def\psubst#1#2{\procpoly{\left[\mathsf{#1}/\mathsf{#2} \right]}}
\newcommand{\lpsubst}[3][q]{\procpoly{\left[\mathsf{#2}/\mathsf{#3} \right]_{\pid{#1}}}}

\newcommand{\cvsubst}[3][q]{\chorpoly{\left[\mathsf{#2}/\mathsf{#3} \right]_{\pid{#1}}}}

\newcommand{\semto}[1][]{\xrightarrow{#1}}
\newcommand{\semtos}[1][]{\xrightarrow{#1}\mathrel{\mkern-14mu}\rightarrow}

\def\Pset{\mathbf{PN}}
\def\Vset{\mathbf{V}}
\def\Eset{\mathbf{E}}

\def\CXset{\mathbf{CV}}
\def\PXset{\mathbf{PV}}
\def\Lset{\mathbf{L}}
\def\aLset{\mathcal{L}}

\def\semLab{\mathbf{Lab}}

\def\abisim{\mathcal R}

\newcommand{\pid}[1]{\mathsf{#1}}
\def\pp{\pid{p}}
\def\pq{\pid{q}}
\def\pr{\pid{r}}
\def\ps{\pid{s}}
\def\pn{\mathsf{pn}}
\def\pnof#1{\pn\left(#1\right)}
\def\apnof#1{\mathsf{apn}\left(#1\right)}
\def\recnof#1{\mathsf{rn}(#1)}

\def\lab{\ell}
\newcommand*{\albl}{\ell}

\newcommand{\tto}{\mathbin{\boldsymbol{\rightarrow}}}

\def\chorpoly#1{\bclr{{#1}}}
\def\seq{\mathbin{\chorpoly{;}}}
\newcommand{\cC}[1][]{\chorpoly{\mathsf{C_{#1}}}}
\newcommand{\cCp}[1][]{\chorpoly{\mathsf{C'_{#1}}}}
\newcommand{\cCpp}[1][]{\chorpoly{\mathsf{C''_{#1}}}}
\newcommand{\cB}[1][]{\cCpp}

\newcommand{\cX}[1][]{\chorpoly{\mathsf{X}_{#1}}}

\def\cnil{\chorpoly{\mathbf{0}}}

\newcommand{\com}[4]{{\pid{#1}.{#2} \tto \pid{#3}.{#4}}}
\newcommand{\choi}[4]{{\pid{#1}.{#2} \tto \pid{#3}.{\Set{#4}}}}

\newcommand{\ccom}[4]{\chorpoly{\com{#1}{#2}{#3}{#4}}}
\newcommand{\cchoi}[4]{\chorpoly{\choi{#1}{#2}{#3}{#4}}}

\newcommand{\gencom}{\com peqx}
\newcommand{\genchois}{\choi{p}{\aLset}{q}{\lab : \cC[\lab] \mid \lab \in \aLset'}}

\newcommand{\genchoi}{\choi{p}{\set{\lab_i}}{q}{\lab : \cC[\lab] \mid \lab \in \aLset'}}
\newcommand{\genchoiEXT}{\choi{p}{\set{\lab_i}}{q}{\lab : \cC[\lab] \mid \lab \in \aLset}}

\def\cgencom{\chorpoly{\gencom}}

\def\cgenchois{\chorpoly{\genchois}}

\newcommand{\genchoisp}{\chorpoly{\choi{p}{\aLset}{q}{\lab : \cCp[\lab] \mid \lab \in \aLset'}}}
\def\cgenchoi{\chorpoly{\genchoi}}
\def\cgenchoiEXT{\chorpoly{\genchoiEXT}}

\newcommand{\procDef}[3][]{{{#2}}_{#1} \triangleq {{#3}}_{#1}}
\newcommand{\gencProcDef}[1][]{\procDef[\chorpoly{#1}]{\cX}{\cC}}

\newcommand{\genpProcDef}[1][]{\procDef[#1]{\pX}{\pP}}
\newcommand{\genfProcDef}[1][]{\fX \triangleq F_\fX}
\newcommand{\procC}{\mathcal{C}}

\def\ctens{\chorpoly{\; \otimes \; }}

\newcommand{\unfoldr}{\mathsf{unfold}}
\newcommand{\ldefr}{\unfoldr}

\def\cto#1{\xrightarrow{\rclr{#1}}}
\def\ctos#1{\xrightarrow{\rclr{#1}}\mathrel{\mkern-14mu}\rightarrow}

\newcommand{\tocom}[4]{\pid{#1}.{#2} \tto \pid{#3}.{#4}}
\newcommand{\tochoi}[2][\lab_i]{\pid{#2}.{#1}}
\newcommand{\tolab}[3][\lab_i]{\pid{#2}.{#1} \tto \pid{#3}}

\def\comrule{\mathsf{com}}
\def\choirule{\mathsf{choice}}

\def\tensrulel{\mathsf{tensor\mhyphen left}}
\def\tensruler{\mathsf{tensor\mhyphen right}}
\def\labrule{\mathsf{label}}

\def\delcomrule{\mathsf{delay\mhyphen com}}
\def\delchoirule{\mathsf{delay\mhyphen choice}}
\def\unfoldrule{\mathsf{proc}}

\def\netpoly#1{\nclr{\mathsf{#1}}}

\newcommand{\nN}[1][]{\netpoly{N_{#1}}}
\newcommand{\nNof}[2][]{\procpoly{\mathsf{N_{#1}}(#2)}}
\newcommand{\nMof}[2][]{\procpoly{\mathsf{M_{#1}}(#2)}}
\newcommand{\nNpof}[2][]{\procpoly{\mathsf{N^{\raisebox{0.4ex}{$\scriptscriptstyle\prime$}}_{#1}}(#2)}}

\newcommand{\nM}[1][]{\netpoly{M_{#1}}}
\newcommand{\nNp}[1][]{\netpoly{N^{\raisebox{0.4ex}{$\scriptscriptstyle\prime$}}_{#1}}}
\newcommand{\nNpp}[1][]{\netpoly{N^{\raisebox{0.4ex}{$\scriptscriptstyle\prime\prime$}}_{#1}}}
\newcommand{\nMp}[1][]{\netpoly{M^{\raisebox{0.4ex}{$\scriptscriptstyle\prime$}}_{#1}}}

\def\nnil{\netpoly{\mathbf{0}}}

\def\procpoly#1{\gclr{\mathsf{#1}}}
\newcommand\pP[1][]{\procpoly{P_{#1}}}
\newcommand\pPp[1][]{\procpoly{P_{#1}^{\raisebox{0.4ex}{$\scriptscriptstyle\prime$}}}}
\newcommand\pQ[1][]{\procpoly{Q_{#1}}}

\newcommand\pQp[1][]{\procpoly{Q_{#1}^{\raisebox{0.4ex}{$\scriptscriptstyle\prime$}}}}
\newcommand{\pX}[1][]{\procpoly{X_{#1}}}
\newcommand{\pY}[1][]{\procpoly{Y_{#1}}}
\newcommand{\pZ}[1][]{\procpoly{Z_{#1}}}

\def\pseq{\procpoly{\mathsf{.}}}

\newcommand{\procP}{\mathcal{P}}

\newcommand{\plab}[3][\albl]{\{#1 : #3\}_{#1 \in #2}}

\newcommand{\pdo}[2][p]{\nclr{\pid{#1}\left[ #2 \right]}}
\def\ppar{\;|\;}
\def\pnil{\procpoly{\mathbf{0}}}

\def\psend#1#2{\procpoly{\pid{#1}!{#2}}}
\def\precv#1#2{\procpoly{\pid{#1}?{#2}}}
\newcommand{\matpsendl}[3][]{\procpoly{\pid{#2}!\set{#3}_{#1}}}
\newcommand{\matprecvl}[3][]{\procpoly{\pid{#2}?\set{#3}_{#1}}}

\newcommand{\psendl}[3][\aLset]{\procpoly{\pid{#3}!\plab{#1}{#2_{\lab}}}}
\newcommand{\pselectl}[3][\lab_i]{\procpoly{\pid{#3}!\{#1 : #2_{#1}\}}}
\newcommand{\precl}[3][\aLset]{\procpoly{\pid{#3}?\plab{#1}{#2_{\lab}}}}
\def\pgensendl{\psendl{\pP}{q}}
\def\pgenrecvl{\precl{\pQ}{p}}
\def\pgenselectl{\pselectl{P}{q}}
\newcommand{\predtree}[2][T]{\mathcal{#1}_{#2}}

\def\pto#1{\xrightarrow{\rclr{#1}}}
\def\ptos#1{\xrightarrow{\rclr{#1}}\mathrel{\mkern-14mu}\rightarrow}
\newcommand{\psendact}[3]{\mathsf{snd}_{\pid{#1},\pid{#2}}(#3)}
\newcommand{\precvact}[3]{\mathsf{rcv}_{\pid{#1},\pid{#2}}(#3)}
\newcommand{\pchoicel}[2]{\pid{#1}.{#2}}
\newcommand{\pselact}[3]{\mathsf{sel}_{\pid{#1},\pid{#2}}(#3)}
\newcommand{\pbraact}[3]{\mathsf{bra}_{\pid{#1},\pid{#2}}(#3)}

\def\pcomrule{\mathsf{com}}
\def\pchoirule{\mathsf{choi}}
\def\plabrule{\mathsf{lab}}

\def\psendrule{\mathsf{send}}
\def\precvrule{\mathsf{recv}}
\def\pchoisrule{\mathsf{choi}}
\def\pselectrule{\mathsf{sel}}
\def\pbranchrule{\mathsf{branch}}
\def\pprocrule{\mathsf{proc}}

\def\pnetrule{\mathsf{par}}

\newcommand{\EPP}[2][]{\mathsf{EPP}_{\pid{#1}}\left(#2\right)}

\def\peq{\approx}
\newcommand{\proves}[1][]{\mathord{\vdash_{#1}\,}}
\newcommand{\dD}{\mathcal{D}}

\def\XS{\mathsf{X}}

\def\fX{\mathsf X}
\def\fY{\mathsf Y}
\newcommand{\fof}[1]{\left\llbracket#1\right\rrbracket}
\newcommand{\fofclean}[1]{\left\llbracket#1\right\rrbracket^{-}}

\def\MALL{\mathsf{MALL}}

\def\PIL{\mathsf{PiL}}

\def\uPIL{\PIL_{\fprocset}}
\def\iPIL{\PIL^{\infty}_{\fprocset}}

\def\ChorL{\mathsf{ChorL}_{\fprocset}}
\def\gChorL{\ChorL^{\mathsf{prog}}}
\def\nChorL{\ChorL^{\mathsf{nice}}}

\def\cneg#1{\overline{#1}}

\def\chan#1#2{\mathsf c_{\pid #1\pid #2}}
\def\lFa#1#2{\forall{#1}.#2}
\def\lEx#1#2{\exists{#1}.#2}

\def\lExp#1#2{\lEx{#1}{\!\left( #2 \right)}}

\newcommand{\multiprem}[2]{\vlhy{\left\{\vlderivation{#1}\right\}_{#2}}}
\makeatletter

\DeclareRobustCommand{\btleft}{\mathrel{\mathpalette\btlr@\blacktriangleleft}}
\DeclareRobustCommand{\btright}{\mathrel{\mathpalette\btlr@\blacktriangleright}}
\newcommand{\btlr@}[2]{%
  \begingroup
  \sbox\z@{$\m@th#1\triangleright$}%
  \sbox\tw@{\resizebox{1.1\wd\z@}{1.1\ht\z@}{\raisebox{\depth}{$\m@th#1\mkern-1mu#2$}}}%
  \ht\tw@=\ht\z@ \dp\tw@=\dp\z@ \wd\tw@=\wd\z@
  \copy\tw@
  \endgroup
}

\def\alphaeq{\equiv_{\mathrel{\alpha}}}

\def\fprocset{\mathfrak{F}}
\def\cprocset{\chorpoly{\mathfrak{F}}}
\def\lsend#1#2{\langle#1\oc#2\rangle}
\def\lrecv#1#2{(#1\wn#2)}
\def\lprec{\btleft}
\newcommand{\biglbra}[1][]{\bigwith_{#1}}
\newcommand{\biglsel}[1][]{\bigoplus_{#1}}
\newcommand{\Biglbra}[1][]{\bigwith\limits_{#1}}
\newcommand{\Biglsel}[1][]{\bigoplus\limits_{#1}}

\def\axrule{\mathsf {ax}}
\def\AXrule{\mathsf {AX}}
\def\cutr{\mathsf {cut}}

\def\mixr{\mathsf{mix}}

\newcommand{\rrule}[1][]{\mathsf{r}^{#1}}

\def\free{\mathsf{free}}
\def\freeof#1{\free(#1)}

\def\init{\mathsf{Init}}
\newcommand{\Cmixr}[1][]{\underset{#1}{\mathsf{C}\mhyphen\mixr}}
\newcommand{\Ccomr}[1][]{\underset{#1}{\mathsf{C}\mhyphen\pcomrule}}
\newcommand{\Cchor}[1][]{\underset{#1}{\mathsf{C}\mhyphen\pchoirule}}
\newcommand{\Clabr}[1][]{\underset{#1}{\mathsf{C}\mhyphen\plabrule}}
\newcommand{\Cbrar}[1][]{\underset{#1}{\mathsf{C}\mhyphen\mathsf{bra}}}
\newcommand{\Crec}[1][]{\underset{#1}{\mathsf{C}\mhyphen \mathsf{unfold}}}
\newcommand{\Caxr}[1][]{\underset{#1}{\mathsf{C}\mhyphen\init}}
\newcommand{\Cwr}[1][]{\underset{#1}{\mathsf{C}\mhyphen\mathsf{clean}}}

\def\chorpol#1{\bclr{#1}}

\newcommand{\chRCof}[1][\cdot]{\chorpol{\mathsf{Extr}{\left( {\color{black} #1}\right)}}}

\def\wbisim{\bisim}
\def\bisim{\sim}

\newcommand{\DER}[1]{\bclr{\mathsf{Dr}\left({\color{black} #1}\right)}}
\newcommand{\DofC}[1]{\bclr{\mathsf{Pr}\left({\color{black} #1}\right)}}
\newcommand{\TofD}[1]{\bclr{\mathsf{Tr}\left({\color{black} #1}\right)}}

\theoremstyle{definition}
\newtheorem{notation}{Notation}
\newtheorem{remark}{Remark}

\begin{document}

\title[Two Sides of the Same Proof: Extraction and Projection in Choreographic Programming]
{Two Sides of the Same Proof: Extraction and Projection in Choreographic Programming}
\author{Matteo Acclavio}
\author{Giulia Manara}
\author{Fabrizio Montesi}
\author{\\Marco Peressotti}
\affiliation{%
  \institution{\\FORM, University of Southern Denmark}
  \city{Odense}
  \country{Denmark}
}

\renewcommand{\shortauthors}{Acclavio et al.}

\begin{abstract}
  A new approach to the theory of choreographic languages has recently been proposed, based on the interpretation of choreographies with finite behaviors as proofs in an extension of linear logic.
  This paper further develops this approach to cyclic proofs to be capable of representing infinite behaviors, and prove a completeness result for choreographies with respect to the class of lock-free networks.

  We show that in a choreographies-as-proofs correspondence, the proof equivalence induced by rule permutations corresponds to bisimilarity, and that, under adequate assumption on the procedure definitions, the proof search procedure has a finite search space, and can be used to extract a choreography from a network.
\end{abstract}

\ccsdesc[500]{Theory of computation~Process calculi}
\ccsdesc[500]{Theory of computation~Linear logic}
\ccsdesc[500]{Theory of computation~Operational semantics}


\maketitle

\section{Introduction}


Choreographic programming (CP) is a programming paradigm for concurrent and distributed systems, where the programmer specifies a global description of the system, called \emph{choreography}, and then the compiler produces a \defn{network} of local programs that implement the choreography using a procedure called \emph{endpoint projection} (EPP) \citep{M13:phd,montesi:book}.
Correctness of EPP establishes a behavioral equivalence, usually in terms of strong bisimulation, between choreographies and their compiled networks produced by EPP.
The advantage of CP is that global descriptions of the system avoids the occurrence of communication problems, such as deadlocks and race conditions, and therefore a network obtained through EPP is guaranteed to be free from such errors.

\begin{figure}
	\includegraphics[width=.8\textwidth]{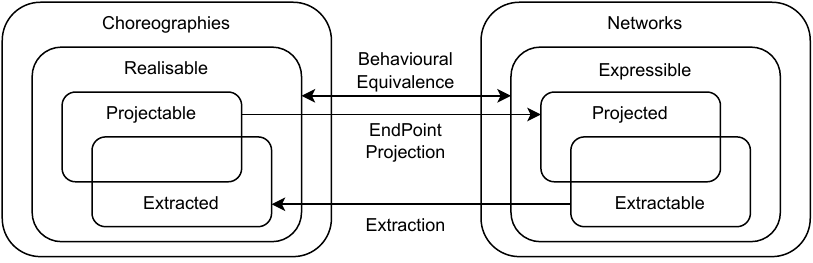}
	\caption{Realisability and expressivity of choreographies and networks.}
	\label{fig:diagram}
\end{figure}

The theory of CP revolves around two central notions: \emph{realisability} and \emph{expressivity}.
\begin{itemize}
	\item Realisability regards the question of \emph{which choreographies can be implemented by networks}.
	This question is usually answered constructively by designing EPP procedures that, given a choreography, produce a network behaviorally equivalent of the choreography.

	\item Expressivity regards the question of \emph{which networks can be specified by choreographic languages}.
	A constructive way to answer this question is to design \emph{extraction} procedures that, given a network satisfying certain properties, produce a choreography behaviorally equivalent to the network.
\end{itemize}
Ideally, these two directions should form two facets of the same underlying correspondence. 
The diagram in~\cref{fig:diagram} summarises this relationship: 
realisable choreographies correspond to expressible networks via behavioural equivalence, 
the classes of projectable choreographies and projected networks identify the fragments connected by projection and the classes of extracted choreographies and extractable networks identify the fragments connected by extraction, respectively. 

Most work on choreographic programming has focused on the first question. 
This is partly because endpoint projection gives a concrete compilation problem with immediate practical relevance, and partly because it is technically more tractable. 
The tractability stems from the fundamental difference the information directly available in networks and choreographies: at the network level, the overall behavior of the system emerges from local programs that describe the behavior of each endpoint separately, whereas at the choreographic level the overall behavior is explicit, with choreographic programs directly describing the intended communications and the recursive structure of infinite behaviors. 
Extraction therefore faces the more challenging task of reconstructing emergent global behavior from distributed local code, a notoriously difficult problem \cite{LLLG16}.

The few works that approach the expressivity question in CP \cite{montesi:cruz:al:choreo_extraction,montesi:cruz:extraction_spwaning,cruz:montesi:safina:implementig_extraction} (or in the adjacent field of Multiparty Session Types~\cite{yoshida:al:extraction_mpst}) provide extraction procedures and establish their correctness with ad hoc arguments, without a direct correspondence with projection. 
This leaves open a central question: \emph{what is the precise relationship between realisability and expressivity}, that is, between networks that realise choreographies and networks expressible as choreographies, or, dually, between realisable choreographies and choreographies that express networks?

\citet{acc:man:mon:ESOP25} made a first step towards answering the question by introducing a new perspective, called \emph{choreographies-as-proofs}, which replaces algorithmic constructs used in previous attempts by proof-theoretic ones. 
As the name suggest, this approach is based on the interpretation of choreographies as proofs in an extension of linear logic, using a \emph{process-as-formulas} interpretation of processes, similar what studied in \cite{miller:pi,mil:tiu:nabla,bru:02,ola:pim:viv:subexp,pimentel:procAsFor,hor:tiu:privateNames,miller:book}, where the behavior of a process is represented as a formula in linear logic, and the operational semantics of processes is represented as a proof system for the formulas, following the logic-programming paradigm.
The core outcome of this approach is the possibility to establish a correspondence between some special representations of the LTSs of a network, called \emph{execution trees}, and derivations in a sequent calculus extending multiplicative additive (first-order) linear logic.
Then, the extraction of a choreography from a network is obtained by performing a proof search with some restrictions on the search space, and the choreography naturally emerges from the structure of the proof.
However, the results are limited to choreographies with finite behaviors.

\subsection{Contribution of the paper}
In this work, we show that the choreographies-as-proofs approach can be extended to choreographies with infinite behaviors. In particular, we show that it can be used not only to design a complete extraction procedure for tail-recursive lock-free networks, but also to establish a new proof of endpoint projection.
These two results show that the class of choreographies with infinite behaviors is expressive enough to capture the class of tail-recursive lock-free networks, solving an open problem in the theory of choreographic programming, and that endpoint projection is complete with respect to this class of choreographies.
Together, these results place the theory of choreographic programming on a solid logical foundation based on the proof theory of linear logic, providing a unified and intuitive framework in which notions of realizability and expressivity can be studied in terms of proof-theoretic notions of derivability and proof search.

For this purpose, we rely on recent development in the proof theory of cyclic proofs (e.g., \cite{santocanale:circular,afs:lei:men:mu,das:pous:nonWF,das:gir:TCL,mir:stu:zen:coalgebraic,kloi:dalm:ven:CPDL,acc:cur:gue:CSL24,das:de:omega}) and in particular of the extension of multiplicative additive linear logic with fixpoints ($\mu\MALL$) \cite{baelde:fix,BaeldeM07,bae:dou:sau:16,das:de:sau:muMALL,de:jaf:sau:phase,nol:sau:tas:local}, to extend the work in \cite{acc:man:mon:ESOP25} with rules capturing (greatest) fixpoints, required to model infinite behaviors.
We introduce a novel notion of \emph{\fairness} for proof, providing the proof-theoretic counterpart of the homonymous notion used in the theory of process calculi stronger than the notion of \emph{progressiveness} used in the theory of non-wellfounded proofs,
which is sound and complete with respect to the class of sequents encoding lock-free tail-recursive networks.
We then define the system $\ChorL$ of synthetic rules (in the sense of \cite{andreoli:focSeq,miller:uniform,ola:pim:concurrentFocussing}), which allow us to restrict the proof search space while preserving completeness.
Fair proofs in $\ChorL$ are in bijection with choreographies, and the proof search procedure can be used to extract a choreography from a network.

The method we propose allows us to extract a choreography from the following network,
\begin{equation}\label{ex:intro:extraction:network}
	\pdo[p]{\psend qe \seq \pX}
	\ppar
	\pdo[q]{\precv px\seq \precv px \seq \precv ry \seq \pY}
	\ppar
	\pdo[r]{\psend qe \seq \pZ}
	\qqwith
	\begin{cases}
		\procDef{\pX}{\psend qe \seq \pX}
	\\
		\procDef{\pY}{\precv px \seq \precv ry \seq \pY}
	\\
		\procDef{\pZ}{\psend qe \seq \pZ}
	\end{cases}
\end{equation}
by performing a proof search in a restricted proof-search space in $\ChorL$ for a sequent encoding the network in \eqref{ex:intro:extraction:network}, obtaining the following cyclic proof (where the $\star$ labels two rules marking the roots of two isomorphic subtrees of the proof).
\begin{equation}\label{ex:intro:extraction:proof}
	\vlderivation{
		\vlin{\Ccomr[\ccom peqx]}{\star}{
			\fof{
				\pdo[p]{\psend qe \seq \pX}
				\ppar
				\pdo[q]{\precv px\seq \precv px \seq \precv ry \seq \pY}
				\ppar
				\pdo[r]{\psend qe \seq \pZ}
			}
		}{
			\vlin{\ldefr}{}{
				\fof{
					\pdo[p]{\pX}
					\ppar
					\pdo[q]{\precv px \seq \precv ry \seq \pY}
					\ppar
					\pdo[r]{\psend qe \seq \pZ}
				}
			}{
				\vlin{\Ccomr[\ccom peqx]}{}{
					\fof{
						\pdo[p]{\psend qe \seq \pX}
						\ppar
						\pdo[q]{\precv px \seq \precv ry \seq \pY}
						\ppar
						\pdo[r]{\psend qe \seq \pZ}
					}
				}{
					\vlin{\Ccomr[\ccom reqy]}{}{
						\fof{
							\pdo[p]{\pX}
							\ppar
							\pdo[q]{\precv ry \seq \pY}
							\ppar
							\pdo[r]{\psend qe \seq \pZ}
						}
					}{
						\vlin{\Crec[\cX]}{}{
							\fof{
								\pdo[p]{\pX}
								\ppar
								\pdo[q]{\pY}
								\ppar
								\pdo[r]{\pZ}
							}
						}{\vlin{}{\star}{
							\fof{
								\pdo[p]{\psend qe \seq \pX}
								\ppar
								\pdo[q]{\precv px\seq \precv px \seq \precv ry \seq \pY}
								\ppar
								\pdo[r]{\psend qe \seq \pZ}
							}
						}{\vlhy{}}}
					}
				}
			}
		}
	}
\end{equation}
Then, from the structure of the proof in \eqref{ex:intro:extraction:proof}, we can extract the following choreography, which we can prove be bisimilar to the initial network in \eqref{ex:intro:extraction:network}.
\begin{equation}\label{ex:intro:extraction:chore}
	\cC=\ccom peqx \seq \ccom peqx \seq \ccom qery \seq \cX
	\qqwith
	\gencProcDef
\end{equation}

The same framework, allows us for a new proof of the endpoint projection theorem, by showing that if we can translate a choreography into a fair proof in $\ChorL$, then the conclusion of the proof is a sequent encoding a network that is bisimilar to the choreography.
Interestingly, this approach allows us to project more choreographies than the ones that can be projected by the endpoint projection algorithms in the literature.

For example, the following choreography is not projectable by the standard endpoint projection algorithm (see, e.g., \cite{montesi:book}), because of the absence of communication of the \emph{knowledge of choice} \cite{CDP12} about the choice made by $\pp$ to $\pr$.
\begin{equation}\label{ex:intro:unprojectable:choreo}
	\cC= \cchoi p{\set{L,R}}q{
	\begin{array}{ll}
		L :& \com\pp e \pr x \seq \cnil
	\\
		R :& \cX
	\end{array}
	}
	\qqwith 
	\gencProcDef
\end{equation}
However, the corresponding proof $\dD$ does not require the knowledge of choice to be transmitted, because the two branches of the $\Cbrar$-rule are \emph{uniform} in their context, meaning that the non-active formulas in both branches are the same, in the same way as for the rule of the additive conjunction in linear logic.
\begin{equation}\label{ex:intro:unprojectable:proof}
	\scalevlderivation[.85]{
		\vliin{\Cbrar[\cchoi p{\set{L,R}}q{\cdots}]}{\star}{
			\fof{
				\pdo[p]{\matpsendl{q}{L: \psend re, R: {\pX[\pp]}}}
				\ppar
				\pdo[q]{\matprecvl{q}{L: \pnil, R: {\pX[\pq]}} }
				\ppar
				\pdo[r]{\precv px} 
			}
		}{
			\vlin{\Ccomr[\ccom perx]+(2\times\Cwr)}{}{
				\fof{
					\pdo[p]{\psend re}
					\ppar
					\pdo[q]{\pnil}
					\ppar
					\pdo[r]{\precv px} 
				}
			}{
				\vlin{\Caxr[\cnil]}{}{\vdash \fof{\nnil}}{\vlhy{}}
			}
		}{
			\vlin{\Crec[\gencProcDef]}{}{
				\fof{
					\pdo[p]{\pX[\pp]}
					\ppar
					\pdo[q]{\pX[\pq]}
					\ppar
					\pdo[r]{\precv px} 
				}
			}{\vlin{}{\star}{
				\fof{
					\pdo[p]{\matpsendl{q}{L: \psend re, R: {\pX[\pp]}}}
					\ppar
					\pdo[q]{\matprecvl{q}{L: \pnil, R: {\pX[\pq]}} }
					\ppar
					\pdo[r]{\precv px}
				}
			}{\vlhy{}}}
		}
	}
\end{equation}
This, allows us to define a more flexible endpoint projection procedure, which can project choreographies that would have required amendment (in the sense of \cite{DLMZ13,CM20,CM23}) in the standard approach.

\subsection{Structure of the paper}
In \Cref{sec:background} we recall some background notions on labeled transition systems and bisimilarity, to fix the notation we use in the paper.
In \Cref{sec:networks} we recall the syntax and semantics of networks, and we introduce the notion of execution tree, which is a special representation of the LTS of a network.
In \Cref{sec:chor} we introduce the syntax and semantics of choreographies, and we recall the standard endpoint projection theorem.
In \Cref{sec:pilu} we extend the system $\PIL$ with rules for fixpoints, and we study its theory.
In \Cref{sec:proofsAsExecutions} we introduce the system $\ChorL$ of synthetic rules, and we show that fair proofs in $\ChorL$ are in bijection execution trees of lock-free tail-recursive networks.
In \Cref{sec:proofsAsChoreos} we use this correspondence to define extraction and projection procedures, and we prove their correctness and completeness.
We conclude in \Cref{sec:conc} with a discussion of related and future work.

\section{Background}\label{sec:background}

In this section we recall the standard notion of (strong) bisimulation on labeled transition systems (LTSs), 
which we use throughout the paper to relate choreographies and networks under a common notion of behavioral equivalence.

A \defn{labeled transition system} is a triple $(S, \Sigma, \semto)$ where
$S$ is a set of \emph{states}, $\Sigma$ is a set of \emph{labels}, and
${\semto} \subseteq S \times \Sigma \times S$ is a \emph{transition relation}.
As customary, we write $s \semto[\mu] s'$ for $(s,\mu,s') \in {\semto}$, where
$\mu$ ranges over $\Sigma$.
A \defn{reduction sequence} is a sequence of states $s_0,\ldots,s_{n+1}$ such that $s_i \semto[\mu_i] s_{i+1}$ for all $i\in\intset0{n}$, in which case we may write $s_0 \semtos[\mu_0\cdots\mu_{n}] s_{n+1}$ (or simply $s_0\semtos s_{n+1}$ if labels are irrelevant or clear from the context).
We write $s \nrightarrow$ there is no $s'$ such that $s\semto[\mu] s'$ for some $\mu\in\Sigma$.

The choreographies and networks defined in \Cref{sec:networks,sec:chor} are both equipped with label transition systems, 
each consisting of the terms of the corresponding language, a set of labels, and a transition relation induced by its operational semantics.

We recall the standard notion of bisimulation between labelled transition systems.
\begin{definition}\label{def:bisim}
	Let $(S, \Sigma, \semto)$ be an LTS.
	A binary relation $\abisim \subseteq S \times S$ is a \defn{bisimulation} if for each pair $(s_1,s_2)\in\abisim$
	and each $\mu \in \Sigma$:
	\begin{itemize}
		\item\label{def:bisim1} if $s_1 \semto[\mu] s_1'$, then there is a $s_2'$ such that $s_2 \semto[\mu] s_2'$ and $(s_1',s_2')\in\abisim$;
		\item\label{def:bisim2} if $s_2 \semto[\mu] s_2'$, then there is a $s_1'$ such that $s_1 \semto[\mu] s_1'$ and $(s_1',s_2')\in\abisim$.
	\end{itemize}
	Two states $s_1$ and $s_2$ are \defn{bisimilar}, written $s_1 \bisim s_2$, if there is a bisimulation $\abisim$ such that $(s_1,s_2)\in\abisim$.
\end{definition}

\section{Network Calculus and Reduction Trees}\label{sec:networks}

In this section, we introduce the \emph{network calculus}, a process calculus for communicating parallel processes that serves as the target language 
of the endpoint projection of choreographies.
We assume a synchronous communication model, where each pair of processes communicate directly using a unique private communication channel; this features gave 
\emph{race freedom} by design, which is a crucial property we need to prove the completeness of our extraction procedure.
We also restrict our attention to \emph{tail-recursive} processes and networks, i.e., those whose procedure calls occur only in tail position. 
This restriction ensures that every network has a finite set of reachable states.

\subsection{Processes}\label{subsec:networks:processes}
We define \defn{process terms} from given pairwise disjoint countable sets of 
\defn{process names} ($\Pset$), 
\defn{variables} ($\Vset$), 
\defn{expressions} ($\Eset$), 
\defn{process variables} ($\PXset$), and 
\defn{label names} ($\Lset$) 
using the following grammar:
\begin{equation}\label{eq:proc-syntax}
	\begin{array}{c}
		\begin{array}{r@{\quad}r@{\;}c@{\;}l}
        \text{\defn{Process terms}}& \pP, \pQ, \pP[\lab],\pQ[\lab] & \Coloneqq &
            \pnil               		\mid 
            \psend{q}{e}  \pseq \pP   	\mid
			\precv{p}{x}  \pseq \pQ   	\mid
			\pgensendl 					\mid
			\pgenrecvl 					\mid
            \pX 				    
    \end{array}
	 \\[5pt]
        \text{where }
        \pid p, \pid q \in \Pset
        \qomma
        x\in\Vset 
        \qomma
        e\in \Eset
        \qomma
        \albl\in\Lset 
        \text{\;, and\;}
        \emptyset \neq \aLset \subseteq \Lset
	\end{array}
\end{equation}
The set of \defn{active process names} in a process term $\pP$ is the set of process names that appear in $\pP$ before any procedure call.
That is, we define $\apnof{\pP}$ as follows:
\begin{equation}
	\begin{array}{c}
		\apnof{\pnil}   =\apnof{\pX} = \emptyset 
	\qquad
		\apnof{\psend{p}{e} \pseq \pP} = \apnof{\precv{p}{x} \pseq \pP} = \set{\pid p} \cup \apnof{\pP}
	\\[10pt]
		\apnof{\pgensendl} = \apnof{\pgenrecvl} = \bigcup_{\lab \in \Lset'}\apnof{\pQ[\lab]}
	\end{array}
\end{equation}
A \defn{(process) procedure definition} is a pair $\genpProcDef$ where $\pX \in \PXset$ (said to be \defn{defined} by the procedure definition), and $\pP$ is a process term (called the \defn{body} of the procedure definition). 
A procedure definition it is \defn{guarded} if $\apnof{\pP} \neq \emptyset$.
A set of procedure definitions $\procP$ is \defn{closed} if it contains a procedure definition $\genpProcDef$ for each process variable occurring in a procedure definition in $\procP$.
It is \defn{closed} with respect to a process term $\pP$ if it is closed and minimal with respect to the property of containing all the process variables occurring in $\pP$ and in $\procP$.
In the following, as in standard presentations of process calculi \cite{SW01,S11},  whenever we define a process $\pP$, we assume that we also fix a finite closed set of guarded procedure definitions $\procP=\procP_{\pP}$, which is minimal with respect to the property of containing all the process variables occurring in $\procP$.
Then, the set of  \defn{process names} $\pnof\pP$ in a process term $\pP$ is defined by extending the function $\apnof\pP$ to also include process names occurring in all process terms in $\procP_{\pP}$.

The set of \defn{process names} $\pnof{\pP}$ in a process term $\pP$ is defined by the same inductive clauses as $\apnof{\pP}$, except for process variables. In particular, if $\pX$ is a process variable and $\procDef{\pX}{\pQ}\in\procP$, then $\pnof{\pX}=\pnof{\pQ}$. Thus, $\pnof{\pP}$ also accounts for the process names occurring in the bodies of the procedure definitions reached through recursive calls.

\subsection{Networks}\label{subsec:networks:networks}

A \defn{network} is a function $\nN$ mapping process names in $\Pset$ to process terms with finite support (i.e., $\nNof\pp\neq\pnil$ for only finitely many $\pp\in\Pset$), such that no process name $\pp$ is mapped to a process term mentioning $\pp$ itself (i.e., $\pp \notin \pnof{\nNof\pp}$).
It is \defn{atomic} if there exists a unique $\pp\in\Pset$ with $\nNof{\pp}\neq\pnil$, in which case we write $\nN=\pdo{\pP}$, where $\nNof{\pp}=\pP$.
If $\nN[1]$ and $\nN[2]$ are networks with disjoint supports (i.e., if $\nNof[i]\pp\neq\pnil$ then  $\nNof[3-i]\pp=\pnil$ with $i\in\set{1,2}$ for all $\pp\in\Pset$), their \defn{parallel composition}, is the network $\nN[1] \ppar \nN[2]$ defined as the union of the two functions.
The parallel composition is a commutative monoid for disjoint networks \cite{montesi:book} with neutral element $\nnil$ defined as the function mapping all process names to $\pnil$.
\begin{remark}
	Every network $\nN$ can be written as the (finite) parallel composition of the atomic networks of the form $\pdo[\pp_i]{\nNof{\pp_i}}$, with $\pp_i\in\set{\pp_1,\ldots,\pp_k}$ in the support of $\nN$.
\end{remark}
As for processes, given a network we assume that we also fix a finite closed set of guarded procedure definitions $\procP$, which is minimal with respect to the property of containing all the process variables occurring in $\procP$ and in the processes occurring in $\nN$.

The operational semantics of networks is given by the three reduction steps in the top of \Cref{fig:psemantics} with labels from the following set.
\begin{equation}\label{eq:semLab}
	\semLab=\Set{				\quad
		\tocom peqx 			\qqomma 
		\tolab pq 				\qqomma 
		\pchoicel{p}{\lab_i}	\quad
	}
\end{equation}
The function $\pn$ is extended to these labels by letting
$\pnof{\tocom peqz} = \pnof{\tolab pq} =\set{\pid p, \pid q}$, 
$\pnof{\pchoicel{p}{\lab_i}} = \set{\pid p}$.
Note that, we always have $\pp\neq\pq$ in the above labels because of the definition of networks.

\begin{figure}[t]
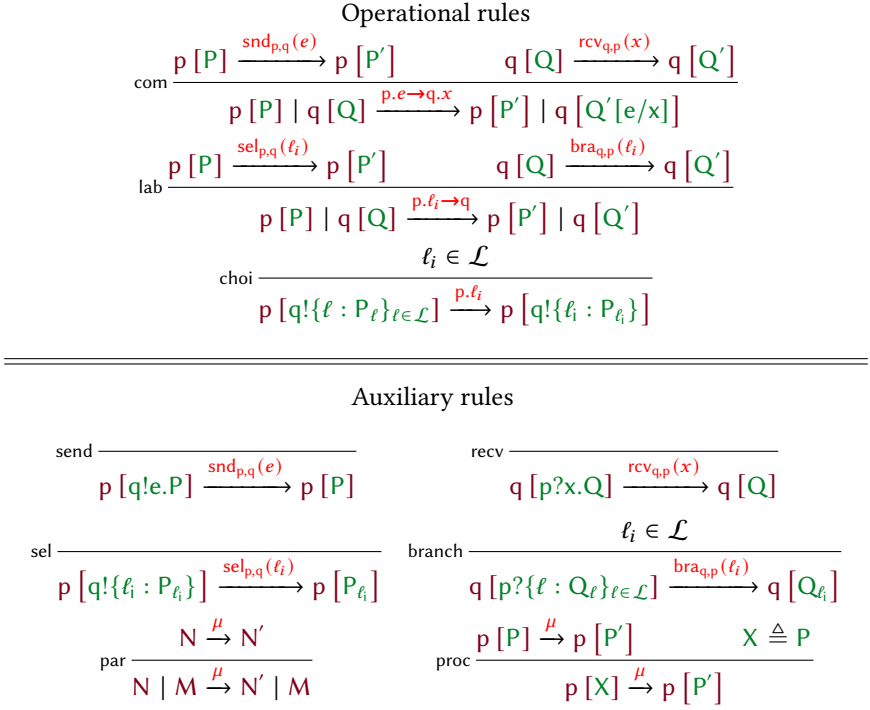

	\adjustbox{max width=\textwidth}{$
	\begin{array}{c}
		\text{Operational rules} \\
		\begin{array}{c}
			\vliiinf{\pcomrule}{}{
				\pdo[p]{\pP} \ppar \pdo[q]{\pQ} \pto{\tocom peqx} \pdo[p]{\pPp} \ppar \pdo[q]{\pQp\psubst{e}{x}}
			}{
				\pdo[p]{\pP} \pto{\psendact{p}{q}{e}} \pdo[p]{\pPp}
			}{\qquad}{
				\pdo[q]{\pQ} \pto{\precvact{q}{p}{x}} \pdo[q]{\pQp}
			}
		\\[15pt]
			\vliiinf{\plabrule}{}{
				\pdo[p]{\pP} \ppar \pdo[q]{\pQ} \pto{\tolab pq} \pdo[p]{\pPp} \ppar \pdo[q]{\pQp}
			}{
				\pdo[p]{\pP} \pto{\pselact{p}{q}{\lab_i}} \pdo[p]{\pPp}
			}{\qquad}{
				\pdo[q]{\pQ} \pto{\pbraact{q}{p}{\lab_i}} \pdo[q]{\pQp}
			}
		\\[15pt]
			\vlinf{\pchoisrule}{}{
				\pdo[p]{\pgensendl} 
				\pto{\pchoicel{p}{\lab_i}} 
				\pdo[p]{\matpsendl q{\lab_i: \pP[{\lab_i}]}}
			}{\lab_i \in \aLset}
		\end{array}
	\\\\ \hline\hline \\[-5pt] \text{Auxiliary rules } \\
		\begin{array}{cc}
			\vlinf{\psendrule}{}{\pdo[p]{\psend qe \pseq \pP} \pto{\psendact{p}{q}{e}} \pdo[p]{\pP}}{}
		&
			\vlinf{\precvrule}{}{\pdo[q]{\precv px \pseq \pQ} \pto{\precvact{q}{p}{x}} \pdo[q]{\pQ}}{}
		\\[15pt]
			\vlinf{\pselectrule}{}{
				\pdo[p]{\matpsendl q{\lab_i: \pP[{\lab_i}]}}
				\pto{\pselact{p}{q}{\lab_i}} 
				\pdo[p]{\pP[{\lab_i}]}
			}{}
		&
			\vlinf{\pbranchrule}{}{\pdo[q]{\pgenrecvl} \pto{\pbraact{q}{p}{\lab_i}} \pdo[q]{\pQ[{\lab_i}]}}{\lab_i \in \aLset}
		\\[15pt]
			\vlinf{\pnetrule}{}{\nN \ppar \nM \pto{\mu} \nNp \ppar \nM}{\nN \pto{\mu} \nNp}
		&
			\vliiinf{\pprocrule}{}{\pdo[p]{\pX} \pto{\mu} \pdo[p]{\pPp}}{\pdo[p]{\pP} \pto{\mu} \pdo[p]{\pPp}}{\qquad}{ \genpProcDef}
		\end{array}
	\end{array}
	$}
	\caption{
		Operational rules for the semantics of networks, and the auxiliary rules required to define them.
	}
	\label{fig:psemantics}
\end{figure}

\begin{definition}\label{def:execution}
	A network $\nN$ is \defn{LF} if for any network $\nNp$ such that $\nN \ptos{} \nNp$, if $\nNpof{\pp}\neq\pnil$ for some $\pp\in\Pset$, then there are networks $\nMp$ and $\nM$ such that $\nNp \ptos{} \nMp \pto{\mu}\nM$ with $\pp\in\pnof\mu$.
\end{definition}

We introduce the notion of \emph{\redtree}, whose branches represent possible executions of a network. 
Whenever a node has multiple successors, these successors correspond to distinct internal choices made by the same process: each successor is obtained by 
selecting a different label among those available to that process.

\begin{definition}\label{def:redtree}
	A \defn{\redtree} of a network $\nN$ is a (possibly infinite) tree $\predtree \nN$ of networks, with nodes networks and with edges by reductions such that:
	\begin{enumerate}
		\item the root of $\predtree \nN$ is the state $\nN$;
		\item if the state $\nM$ has a unique child $\nMp$ then $\nM \pto{\mu} \nMp$ for a   $\mu \in \set{\tocom peqx , \ \tochoi[\lab] p, \ \tolab[\lab] pq \mid \pid p, \pid q \in \Pset,\lab\in\Lset}$;
		\item\label{def:redtree3} if the state $\nM$ has more than one child, then $\nM = \pdo{\pgensendl}\ppar \nM'$ and the set of children of $\nM$ is the set of all states $\nM[\ell]$ such that $\nM \pto{\tochoi[\ell]{p}} \nM[\ell]$ with $\ell \in \aLset$;
		
		\item $\predtree \nN$ is maximal: a state $\nM$ has no children only if $\nM \nrightarrow$.
	\end{enumerate}
	An execution in $\predtree\nN$ is a maximal path $\nN=\nN_0,\ldots, \nN_n,\ldots$. 

	A \redtree $T$ is \defn{\ok} if
	every state $\nM[0]$ occurring in $T$ with $\nMof[0]{\pp}\neq\pnil$,
	has a descendant $\nM[i]$ with a child $\nM[i+1]$
	such that $\nM[0] \ptos{} \nM[i] \pto{\mu} \nM[i+1]$ (for a label $\mu$ with $\pp\in\pnof{\mu}$)
	and
	$\nMof[0]{\pp}=\ldots=\nMof[i]\pp\neq\nMof[i+1]\pp$
\end{definition}

\begin{lemma}\label{lem:generate_redtree}
	Let $\nN$ be a network admitting a \ok \redtree $\predtree\nN$.
	If $\nN \pto{\mu} \nM$ then $\nM$ admits a \ok \redtree $\predtree\nM$.
\end{lemma}
\begin{proof}
If $\nM$ already occurs as a state of $\predtree\nN$, we let $\predtree\nM$ be the subtree of $\predtree\nN$ rooted at $\nM$.
The obtained tree is a \redtree and is \ok since $\predtree\nN$ is.

Otherwise, we construct an auxiliary \redtree rooted at $\nN$ extending the transition
$\nN\pto{\mu}\nM$.
We show that this tree is \ok, and conclude by considering its subtree rooted at $\nM$.
If $\mu\neq\pchoicel{p}{\lab_i}$, the root has $\nM$ as its unique child, and we continue by following the reductions of $\predtree\nN$.
When the transition labeled by $\mu$ is encountered in $\predtree\nN$, we skip it and continue from its successor.
If $\mu=\pchoicel{p}{\lab_i}$, besides $\nM$ we add all the other children required by
\Cref{def:redtree}.\ref{def:redtree3}.
For each such child $\nM[\lab]$, we continue by following the corresponding subtree of $\predtree\nN$.
Whenever the same choice node is reached, only the branch labelled $\lab$ is retained.
Hence, every branch of the new tree is obtained from a branch of $\predtree\nN$ by removing the initial reductions 
involving the processes in $\pnof{\mu}$ that have already been performed at the root.
All remaining reductions are unchanged, and therefore the resulting tree is a \redtree.
The \ok property is preserved: for every state $\nM[0]$ in the new tree and every process $\pp$ with $\nMof[0]{\pp}\neq\pnil$, 
the witnessing reduction required by the \ok condition in $\predtree\nN$ is either preserved in the new tree or coincides with the initial 
reduction $\mu$ already performed at the root.
\end{proof}

\begin{theorem}\label{thm:LF}
	Let $\nN$ be a network.
	Then, $\nN$ is \LF iff it admits a \ok \redtree.
\end{theorem}
\begin{proof}
	If $\nN$ is \LF, then any execution of $\nN$ in a \redtree $\predtree\nN$ satisfies the \LF property, and therefore $\predtree\nN$ is \ok.
	Conversely, assume that $\nN \ptos{} \nNp$. By repeatedly applying \Cref{lem:generate_redtree}, we obtain a \ok \redtree $\predtree{\nNp}$.
	The result follows from the definition of an \ok \redtree.
\end{proof}

\begin{definition}
	A \redtree is \defn{finitely representable} if it can be represented by a finite graph (with edges labels on edges), whose paths are in bijection with the executions of $\predtree \nN$.
	We call such graph the \defn{cyclic representation} of $\predtree \nN$.
\end{definition}

\begin{remark}\label{thm:fintree}
	Since we are considering only tail-recursive networks, the set of reachable states of a network $\nN$ is finite.
	Therefore, any \redtree $\predtree \nN$ of a network $\nN$ is finitely representable.
\end{remark}

\section{Choreographies\dots now with tensor!}\label{sec:chor}

In this section, we fix the syntax and semantics of the choreographic language used in this paper. 
To establish completeness of our extraction procedure, we extend the standard syntax of choreographies with a tensor operator ($\ctens$), 
which composes choreographies with disjoint sets of process names and represents independent behaviors proceeding in parallel.

To get the intuition on why the tensor operation is needed, 
consider the following network, 
where processes $\pp$ and $\pq$ engage in a choice, 
while $\pr$ and $\pid s$ repeatedly communicate independently of that choice:
$$
\pdo[p]{\matpsendl{q}{L: \psend qe \pseq \pnil, R: {\pX[\pp]}}}
\ppar
\pdo[q]{\matprecvl{p}{L: \precv px \pseq \pnil, R: {\pX[\pq]}}}
\ppar
\pdo[r]{\psend sf \pseq \pX[\pr]}
\ppar
\pdo[s]{\precv ry \pseq \pX[{\pid s}]}
$$

where $ {\pX[\pp]} \triangleq \matpsendl{q}{L: \psend qe \pseq \pnil, R: {\pX[\pp]}}$, $ {\pX[\pq]} \triangleq \matprecvl{p}{L: \precv px \pseq \pnil, R: {\pX[\pq]}}$,
$\pX[\pr] \triangleq \psend sf \pseq \pX[\pr]$ and $\pX[{\pid s}] \triangleq \precv ry \pseq \pX[{\pid s}]$.

The behavior of this network decomposes into two fully independent
components: the interaction between $\pp$ and $\pq$, and the interaction
between $\pr$ and $\pid s$. The choreography capturing it must reflect
this separation, as in \eqref{ex:tensor}.
\begin{equation}\label{ex:tensor}
	\cX[1]
	\ctens
	\cchoi{p}{\set{L,R}}{q}{
		\begin{array}{ll}
			L :& \com\pp e \pq x \seq \cnil\!\!
		\\
			R :& \cX[2]
		\end{array}
	}
	\qwith
	\begin{cases}
		\procDef{\cX[1]}{\ccom sfry \seq \cX[1]}
	\\	
		\procDef{\cX[2]}{\cchoi{p}{\set{L,R}}{q}{
		\begin{array}{ll}
			L :& \com\pp e \pq x \seq \cnil\!\!
		\\
			R :& \cX[2]
		\end{array}
	}}
	\end{cases}	\!\!\!
\end{equation}
Without the tensor, any tail recursive choreography over all four processes $\pp$, $\pq$,
$\pr$, $\pid s$ would force their behaviors to share the same recursive variable, conflating independent components into
a single recursive structure.  
This prevents bisimilarity with the
network above: if $\pp$ and $\pq$ chose $L$ and terminated, the
choreography would also force $\pr$ and $\pid s$ to terminate, whereas in
the network they continue communicating independently.

\begin{remark}
	The semantics of choreographies we consider in this paper does not include runtime terms to represent partially unfolded procedure calls, as in \cite{montesi:book}.
	While these runtime terms are in general necessary in the presence of `$\tau$-transitions' in the network semantics to unfold the procedure calls, 
	our network semantics does not include such transitions, and the runtime terms are therefore unnecessary.

	In fact, our extraction procedure benefits from this simplification, as it eases the extraction of the choreography 
	from the network, even when amendment (in the sense of \cite{DLMZ13,CM20,CM23}) is required due to a lack of knowledge-of-choice transmission, 
	as shown in \Cref{ex:extraction}.
\end{remark}

\subsection{Syntax}\label{subsec:chor:syntax}
The set \defn{\chors} is generated from the same pairwise disjoint sets of \defn{process names} ($\Pset$), 
\defn{variables} ($\Vset$), \defn{expressions} ($\Eset$), \defn{choreographic variables} ($\CXset$), and \defn{label names} ($\Lset$) 
used in the definition of networks. Its syntax is given by the following grammar:
\begin{equation}\label{eq:def-chor}
	\begin{array}{l@{\;}c@{\;}l@{\;\vrule\;}l}
		\cC & \coloneqq & \cnil  		& \mbox{terminated choreography} 	\\
		&\mid & \cX						& \mbox{choreographic variable}		\\
		&\mid &\cgencom \seq \cC 		& \mbox{communication}				\\
		&\mid & \cgenchois 				& \mbox{choice}						\\ 
		&\mid & \cC[1] \ctens \cC[2] 	& \mbox{tensor composition}			
	\end{array}
\end{equation}
where $\emptyset\neq\aLset \subseteq \aLset'\subseteq \Lset$ and $\sizeof{\aLset}\in\set{1,\sizeof{\aLset'}}$ and $\pp, \pq \in \Pset$, $\pp\neq\pq$.

The term $\cnil$ denotes a terminated \chor, while $\cX$ denotes a choreographic variable, which we use to allow our semantics to define recursive terms (see below). 
The \chor $\ccom peqx \seq \cC$ starts with an instruction $\ccom peqx$ requiring $\pp$ to send the expression $e$ to $\pq$, and continue as $\cC$, 
where $x$ is replaced by $e$ at $\pq$.
The \chor $\cgenchois$ instructs $\pp$ to choose one label in $\aLset$ and communicate it to $\pq$, to then continue as $\cC[\lab]$ for the selected label $\lab$. Here $\pq$ could, a priori, offer more labels than those that $\pp$ can select, which is why we have $\aLset \subseteq \aLset'$.
The \chor $\cC[1]\ctens\cC[2]$ is the parallel composition of two choreographies.

A \chor is \defn{finite} if it contains no choreographic variables.
The set of \defn{active process names} in a choreography $\cC$ is the set of process names that appear in $\cC$ before any procedure call.
That is, we define $\apnof{\cC}$ as follows:
$$
\begin{array}{c}
	\apnof{\cnil} = \apnof{\cX} = \emptyset  
	\qquad 
	\apnof{\cgencom \seq \cC} = \set{\pid p, \pid q} \cup \apnof{\cC} 
	\\[5pt]
	\apnof{\cC[1] \ctens \cC[2]} = \apnof{\cC[1]} \cup \apnof{\cC[2]}
	\\[5pt]
	\apnof{\cgenchois} = \set{\pid p, \pid q} \cup (\bigcup_{\lab \in \aLset'} \apnof{\cC[\lab]}) 
\end{array}
$$
We define \defn{(choreographic) procedure definitions} similarly to process procedure definitions: they are pairs $\gencProcDef$ where the variable $\cX \in \CXset$ is defined by the choreography $\cC$ (the \defn{body} of the procedure definition). 
The set of \defn{process names} $\pnof\cC$ in a choreography $\cC$ is defined by extending the function $\apnof\cC$ to also include process names occurring in all choreographies occurring in the set of procedure definitions $\procC$.
\footnote{
	In practice, for finite closed sets of procedure definitions, $\pnof{\cX}$ can always be computed for all $\cX$ as shown in \cite{CMP23}.
}
A procedure definition $\gencProcDef$ is \defn{guarded} if{ $\apnof\cC \neq \emptyset$. }
A set of procedure definitions $\procC$ is \defn{closed} if it contains a procedure definition $\gencProcDef$ for each choreographic variable occurring in a procedure definition in $\procC$.

\begin{notation}
	In the following, 
	as in standard presentations of process calculi \cite{SW01,S11}, 
	whenever we define a choreography $\cC$ we assume that we also fix a finite closed set of guarded procedure definitions $\procC$, which is minimal with respect to the property of containing all the choreographic variables occurring in $\cC$
\end{notation}

\begin{definition}
	A \chor $\cC$ is \defn{\chorok} if, for every sub-term of the form $\cC[1] \ctens \cC[2]$, the sets of process names of $\cC[1]$ and $\cC[2]$ are disjoint. 
	That is, $\pnof{\cC[1]} \cap \pnof{\cC[2]} = \emptyset$.
\end{definition}

\begin{remark}
	By definition, a recursive variable cannot occur in both components of a parallel composition in a \chorok \chor.
	Indeed, if $\cX$ occurs in both $\cC[1]$ and $\cC[2]$, then the process names occurring in the body of $\cX$ belong to both $\pnof{\cC[1]}$ and $\pnof{\cC[2]}$. Hence, $\pnof{\cC[1]} \cap \pnof{\cC[2]} \neq \emptyset$.
\end{remark}

\begin{example}
	The choreography
	$\cC[1]= 
	\cX[1]
	{\ctens}
	\cchoi{p}{\set{L,R}}{q}{
		\begin{array}{ll}
			L :& \com\pp e \pr x \seq \cnil
		\\
			R :& \cX[2]
		\end{array}
	}$
	with 
	$\procDef{\cX[1]}{\ccom serx \seq \cX[1]}$
	and
	$\procDef{\cX[2]}{\ccom peqx \seq \cX[2]}$
	is \chorok,
	while the choreography
	$\cC[2] = \cX$, with $\procDef{\cX}{(\ccom peqx \seq \cX )\ctens (\ccom resy \seq\cX)}$ is not \chorok since 
	$\pnof{\ccom peqx \seq \cX } = \pnof{\ccom resy \seq\cX}$.
\end{example}

\subsection{Semantics}\label{subsec:chor:semantics}

In \Cref{fig:semantics}, we define the labelled operational semantics of \chors. Transitions are labelled by elements of $\semLab$, 
defined in \Cref{eq:semLab}. Intuitively, a label records the processes participating in the corresponding step. The semantics extends the standard operational 
semantics of choreographies with tensor composition $\ctens$ quite naturally:
\begin{itemize}
	\item Rule $\comrule$ executes a communication: process $\pp$ sends $e$ to $\pq$, and the choreography continues as $\cC\cvsubst{e}{x}$, that is, the choreography $\cC$ where all occurrences of $x$ at $\pq$ are replaced by $e$.

	\item Rule $\unfoldrule$ unfolds a choreographic variable according to its procedure definition.

	\item Rules $\choirule$ and $\labrule$ implement choices in two steps. First, $\pp$ selects a label $\lab_i$; then, the selected label is communicated to $\pq$,  and the choreography proceeds as $\cC[\lab_i]$. The rules $\tensrulel$ and $\tensruler$ allow either component of a tensor composition to reduce independently. 

	\item Rules $\delcomrule$ and $\delchoirule$ allow a transition in a continuation to be performed before its enclosing communication or choice, provided that  the transition does not involve the processes controlling that prefix. In rule $\delcomrule$, the delayed transition must not involve either $\pp$ or $\pq$.  In rule $\delchoirule$, it must not involve either $\pp$ or $\pq$ and must be enabled in every branch $\cC[\lab]$ for $\lab\in\aLset'$. Thus, the delayed transition is independent of the label eventually selected by $\pp$.
\end{itemize}

\begin{figure}[t]
	\adjustbox{max width=\textwidth}{$
	\begin{array}{c}
		\vlinf{\comrule}{}{\cgencom \seq \cC \ \cto{\tocom peqx} \  \cC \cvsubst{e}{x}}{}
	\qquad\qquad
		\vliiinf{\unfoldrule}{}{\cX \ \cto{\mu} \ \cCp}{\cC \ \cto{\mu} \ \cCp}{\qquad}{\gencProcDef}
	\\[12pt]
		\vlinf{\choirule}{}{\cgenchois \ \cto{\tochoi p} \ \cgenchoi}{ \lab_i \in \aLset}
	\\[12pt]
		\vlinf{\labrule}{}{\cgenchoi \ \cto{\tolab pq} \ \cC[\lab_i]}{ \lab_i \in \aLset}
	\qquad
		\vlinf{\tensrulel}{}{\cC[1] \ctens \cC[2] \ \cto{\mu} \ \cCp[1] \ctens \cC[2]}{ \cC[1] \cto{\mu} \cCp[1]}
	\\[12pt]
		\vliiinf{\delcomrule}{}{\cgencom \seq \cC \ \cto{\mu} \ \cgencom \seq \cCp}{ \cC \ \cto{\mu} \ \cCp}{\qquad}{ \pnof{\mu} \cap \set{\pp , \pq} = \emptyset}
	\qquad
		\vlinf{\tensruler}{}{\cC[1] \ctens \cC[2] \ \cto{\mu} \ \cC[1] \ctens \cCp[2]}{ \cC[2] \cto{\mu} \cCp[2]}
	\\[12pt]
		\vliiinf{\delchoirule}{}{
			\cgenchois \ \cto{\mu} \ \genchoisp
		}{
			\cC[\lab] \ \cto{\mu} \ \cCp[\lab]
			\text{ for all } \lab \in \aLset'
		}{\qquad}{\pnof{\mu} \cap \set{\pp, \pq} = \emptyset}	
	\end{array}
	$}
	\caption{Semantics of choreographies.
	}
	\label{fig:semantics}
\end{figure}

The bisimulation equivalence ($\bisim$) induced by the semantics of \chors on the set of \chors has the following properties.

\begin{proposition}\label{thm:tensor_bisimilarity}
	For all \chors $\cC, \cC[1], \cC[2], \cC[3]$, the following properties hold:
	\begin{itemize}
		\item $\cC[1] \ctens \cC[2] \sim \cC[2] \ctens  \cC[1]$
		\item $\cC[1] \ctens (\cC[2] \ctens \cC[3]) \sim (\cC[1] \ctens \cC[2]) \ctens \cC[3]$
		\item $\cC \ctens \cnil \sim \cC$
	\end{itemize}
\end{proposition}

We recall the definition of lock-freedom and the statement of the lock-freedom holds for \chorok \chors.

\begin{definition}
	A \chor $\cC$ is \defn{\LF} if, for every $\cCp$ such that $\cC \ctos{} \cCp$ and every $\pp\in\pnof{\cCp}$, there exist $n\geq 1$ and a choreography $\cCpp$ such that
	$\cCp \ctos{\mu_1\cdots\mu_n}\cCpp$
	with $\pp\in\pnof{\mu_n}$ and $\pp\notin\pnof{\mu_i}$ for all $i\in\intset1{n-1}$.
\end{definition}

\begin{theorem}\label{thm:coreoLF}
	Let $\cC$ be a \chorok \chor, then $\cC$ is lock-free. 
\end{theorem}
\begin{proof}
	We reason by induction on the structure of $\cC$.
	For finite \chors, the claim is trivial (see, e.g., \cite{montesi:book}).
	Otherwise, we can assume $\cC = \cX[i]$ and $\gencProcDef[i]$ and, by guardedness every process name in $\pnof{\cX[i]}$ is also in $\apnof{\cC[i]}$. 
	Hence, $\cC[i]$ eventually perform a transition involving $\pp$ if and only if the finite choreography obtained from $\cC[i]$ by replacing every choreography variable with $\cnil$ can eventually perform a transition involving $\pp$.
\end{proof}

\section{Cyclic proofs in $\PIL$}\label{sec:pilu}

In this section, we introduce the proof theoretical framework used to establish the completeness of our extraction procedure.
We recall the fragment of the system $\PIL$ introduced in \cite{acc:man:mon:ESOP25} an extension of multiplicative additive linear logic with a non-commutative, non-associative, and self-dual connective $\lprec$, and nominal quantifiers to capture action prefixing and name restriction in process calculi.
However, in this paper, we are not interested in the binding operators of process calculi, therefore we take a fragment of $\PIL$ without nominal quantifiers.

In this work we are interested in reasoning about processes with procedure calls, requiring to introduce a way to capture procedure calls in the logic.
One option could be to use fixpoint operators as in $\mu\MALL$, the extension of multiplicative additive linear logic with least and greatest fixpoint operators \cite{baelde:fix,BaeldeM07,doumane:phd,das:de:sau:muMALL}.
See \Cref{rem:mumall}.
However, we want to reason about processes with procedure calls without introducing the full power of fixpoint operators, and keep the syntax and the semantics of the logic as close as possible to the one of processes with procedure calls.
Therefore, we extend the system $\PIL$ with procedural variables to model procedure calls. 
These could be seen either as second-order propositional variables quantified by an existential quantifier, but where the witness is fixed to be the body of the procedure definition, or as fixpoint formulas of the form $\nu X. A[X]$, where $A[X]$ is the formula encoding the body of the procedure definition.
\footnote{
	Notice that this intuitive fixpoint view uses the greatest fixpoint operator ($\nu$), because programs with procedure calls allow for infinite unfolding, and not only unbounded finite unfolding, which would be captured by the least fixpoint operator ($\mu$).
}

\subsection{Formulas}\label{sec:pilu:forms}

\defn{Formulas} are generated from the set of process names $\Pset$, the set of expressions $\Eset$, the set of label $\Lset$, the set of propositional variables%
\footnote{
	The use of the same notation for the set of propositional variables and the set of process variables an explicit decision:
	we want a bijection between the set of propositional variables and the set of process variables, which we subsume by using the same notation.
}
$\PXset$, and
a set of variables $\Vset$ and the \defn{unit} $\lunit$ using the following grammar:
\begin{equation}\label{eq:formulas}
	A,B \coloneqq 
	\lunit 					\mid 
	\lsend{\chan pq}{\xi}	\mid
	\lrecv{\chan pq}{\xi}	\mid
	A\lpar B				\mid
	A\ltens B				\mid
	A\lprec B				\mid
	A\lplus B				\mid
	A\lwith B				\mid
	\lFa xA					\mid
	\lEx xA					\mid
	\fX						
\end{equation}
with $\chan pq$ a pair of distinct process names $\pp,\pq \in\Pset$, $\xi\in\Vset\cup\Eset\cup\Lset$, $\fX \in \PXset$, and $x\in\Vset$.

A formula is an \defn{atom} if it is either an \defn{output} formula $\lsend{\chan pq}{\xi}$, or an \defn{input} formula $\lrecv{\chan pq}{\xi}$.
We call $\chan pq$ a \defn{channel} with \defn{endpoints} $\pp $ and $\pq$.%
\footnote{
	That is, atoms are generated by three binary predicates: the one constructing channels $\chan--$, the send $\lsend--$ and the receive $\lrecv--$.
	Additionally, we require that the channel constructor only takes as arguments pairs of distinct process names, 
	that $\lsend--$ and $\lrecv--$ only take as first argument channels and as second argument variables or expressions.
}
They can be seen as binary predicates, where the first argument is a channel (made of a pair of process names, representing its endpoints) and the second is the message which is sent or received on that channel.
The connectives $\lpar$ and $\ltens$ \resp{$\lplus$ and $\lwith$} are the multiplicative \resp{additive} disjunction and conjunction of linear logic.
The connective $\lseq$ is a non-commutative, non-associative, self-dual multiplicative connective.
The quantifiers $\forall$ and $\exists$ are the standard first-order quantifiers.
Finally, $\fX$ is a \defn{propositional variable}.

In the following, we may write $\bigoplus_{i=1}^n A_i$ and $\bigwith_{i=1}^n A_i$ as shorthand for $A_1\lplus (A_2 \lplus (\cdots \lplus A_n))$ and $A_1\lwith (A_2 \lwith (\cdots \lwith A_n))$, respectively.
In particular, we may write $\bigwith_{i=1}^1 A_i$ to simply denote $A_1$.

We assume the following de Morgan duality for formulas:
\begin{equation}\label{eq:de_morgan}
	\begin{array}{cccc}
		\lunit^\perp = \lunit
	&
		\lsend{\chan pq}{\xi}^\perp = \lrecv{\chan pq}{\xi}
	&
		\lrecv{\chan pq}{\xi}^\perp = \lsend{\chan pq}{\xi}
	&
		\lneg{(\lneg A)} = A
	\\
		(A\lpar B)^\perp = A^\perp \ltens B^\perp
	&
		(A\lprec B)^\perp = A^\perp \lprec B^\perp
	&
		(A\lplus B)^\perp = A^\perp \lwith B^\perp
	&
		(\lFa xA)^\perp = \lEx xA^\perp
	\end{array}
\end{equation}

The notion of free and bound variables is defined in the standard way, with the quantifiers $\forall$ and $\exists$ binding the variable $x$ in their scope, and the propositional variables $\fX$ being always free.
We denote by $\freeof A$ the set of free variables of a formula $A$.
\begin{gather*}
	\freeof{\lunit} = \freeof{\lsend{\chan pq}{\xi}} = \freeof{\lrecv{\chan pq}{\xi}}= \freeof{\fX}=\emptyset
\\
	\freeof{A\lpar B} = \freeof{A\ltens B} = \freeof{A\lprec B} = \freeof{A\lplus B} = \freeof{A\lwith B} = \freeof{A} \cup \freeof{B}
\\
	\freeof{\lFa xA} = \freeof{\lEx xA} = \freeof{A} \setminus \set{x}
\end{gather*}
	
We consider formulas to be defined up to $\alpha$-equivalence generated by the following clauses:
\begin{equation}
	\begin{array}{r@{\alphaeq}l@{\qquad}l}
		a 	 					& 	a
	&
		\text{if } a\in\set{\lunit, \lsend{\chan pq}{\alpha}, \lrecv{\chan pq}{\alpha}}
	\\
		A_1 \circleddot A_2 	& B_1 \circleddot B_2
	&	
		\text{if $A_1 \alphaeq B_1$, $A_2 \alphaeq B_2$, and $\circleddot \in \set{\lpar,\lprec,\ltens,\lplus,\lwith}$}
	\\
		\lQu x A 				& 	\lQu y A \fsubst yx
	&
		\text{
			if $y$ fresh for $A$, and $\lqusymb \in \set{\forall, \exists}$, 
			and $\circleddot \in \set{\lpar,\lprec,\ltens,\lplus,\lwith}$
		}
	\end{array}
\end{equation}

We can encode process terms as formulas in $\PIL$ 
as follows using the encoding $\fof\cdot$ defined by induction as follows:
\begin{equation}\label{eq:fof}
	\adjustbox{max width=\textwidth}{$
	\begin{array}{r@{\;=\;}l@{\quad}r@{\;=\;}l}
		\fof{\pdo{\pnil}} 	& \lunit
	&
		\fof{\pdo{\pX}} 	& \fX
	\\[5pt]
		\fof{\pdo{\psend qe \pseq \pP}} & \lsend{\chan{p}{q}}{e} \lprec \fof{\pdo{\pP}}
	&
		\fof{\pdo[p]{\pgensendl}} 		& \biglbra[\ell\in\Lset]{\left(\lsend{\chan{p}{q}}{\ell} \lprec \fof{\pdo[p]{\pP[\ell]}} \right)}
	\\[5pt]
		\fof{\pdo[q]{\precv pk \pseq \pP}} 	& \lExp{k}{\lrecv{\chan{q}{p}}{k} \lprec \fof{\pdo[q]{\pP}}}
	&
		\fof{\pdo[q]{\pgenrecvl}} 			& \biglsel[\ell\in\Lset]{\left(\lrecv{\chan{q}{p}}{\ell} \lprec \fof{\pdo[q]{\pP[\ell]}}\right)}
 	\end{array}
	$}
\end{equation}

We encode process terms as formulas in $\PIL$ by the inductively defined encoding $\fof{\cdot}$ given in \Cref{eq:fof}. 
Note that, $\fof{\pdo{\pX}} $ is defined as the identity on $\pX\in\PXset$ being at the same time the set of process variables used to define process terms and the set of propositional variables used to define formulas.

The terminated process is encoded as the multiplicative unit $\lunit$. A process variable is encoded as its corresponding propositional variable. 
Output prefixes are encoded by a $\lprec$-formula with the formula representing the output action on the left and the encoding of the continuation on the right.
Similarly, an input prefix is encoded by a $\lprec$-formula with the formula representing the input action on the left and the encoding of the continuation on the right. However, here the encoding uses the existential quantification to capture the binding of the input variable $k$ in the continuation.
Selections and branchings are encoded as additive choices. 
A process making an internal choice is encoded using the additive conjunction (ranging over the labels in $\aLset$) of $\lprec$-formulas with the formula representing the output of the label on the left and the encoding of the continuation on the right.
Similarly, a process making an external choice is encoded using the additive disjunction (ranging over the labels in $\aLset$) of $\lprec$-formulas with the formula representing the input of the label on the left and the encoding of the continuation on the right.

\begin{remark}\label{rem:local_process:formulas}
	By definition of the encoding of atomic networks, $\fof{\pdo\pP}$ only contains atoms of the form $\lsend{\chan{p}{q}}{\xi}$ or $\lrecv{\chan{p}{q}}{\xi}$ for some channel variables of the form $\chan pq$,
	where the $\pq$ may not necessarily be the same for all atoms in $\fof{\pdo\pP}$, but the process name $\pp$ is the same for all atoms in $\fof{\pdo\pP}$.
\end{remark}

\subsection{Sequent calculi}\label{sec:pilu:sequent}
We assume the reader to be familiar with the general notions of sequent calculus and cyclic proofs, which we recall briefly for the sake of completeness.

A \defn{sequent} is a set of formula occurrences.%
\footnote{
	Occurrences are assumed to be uniquely identified, in contrast to the multiset presentation where formulas carry multiplicities. 
	This convention simplifies the tracking of occurrences in a \derivation.
}
The set $\freeof\Gamma$ of free variables of a sequent $\Gamma$ is the union of the free variables of the formulas in $\Gamma$.

\begin{definition}
	We extend the encoding of `processes as formulas' in \Cref{eq:fof} to `networks as sequents' by letting the encoding of a network be the sequent containing the encoding of each process in the atomic network constituting it. That is,
	\begin{equation}\label{eq:encoding_networks}
		\fof{
			\pdo[p_1]{\pP[1]} \ppar \cdots \ppar \pdo[p_n]{\pP[n]}
		}
		= 
		\fof{\pdo[p_1]{\pP[1]}} , \ldots , \fof{\pdo[p_n]{\pP[n]}}
	\qquand
		\fof{\nnil}=\lunit
	\end{equation}
	We may also write $\fofclean{\nN}$ for the sequent obtained from $\fof{\nN}$ by removing all occurrences of the unit $\lunit$ except when $\fof{\nN}=\lunit$.
	
	Note that the encoding of a network $\nN$ is a set of formulas where all atoms have the same process name $\pp$ appearing in their channel variables, and such $\pp$ is the same for all formulas in $\fof\nN$ (see \Cref{rem:local_process:formulas}).
	We call such process name the \defn{location} of a formula in $\fof\nN$.
\end{definition}

\begin{definition}
	A \defn{sequent rule} $\rrule$ with $n\geq 0$ \defn{premises} $\Gamma_1,\ldots,\Gamma_n$ and \defn{conclusion} $\Gamma$ is written $\vliiinf{\rrule}{}{\Gamma}{\Gamma_1}{\cdots}{\Gamma_n}$ .
	Given a set of rules $\XS$, a \defn{derivation} $\dD$ in $\XS$ of a sequent $\Gamma$ is a tree of sequents, with \defn{conclusion} (i.e., root) $\vdash \Gamma$, where the children of each sequent in $\dD$ are the premises of an instance of a rule in $\XS$ with conclusion the sequent labeling the parent node.
	If a node has no children, then either is the conclusion of a rule with no premises (e.g., rules $\axrule$ and $\lunit$ in \Cref{fig:rules}), or it is called an \defn{open premise} of the derivation.
	A \defn{proof} is a derivation with no open premises.
	A \defn{subderivation} of a derivation $\dD$ is a subtree $\dD'$ of  $\dD$.

	Two derivations are \defn{isomorphic} if there is a bijection between their nodes preserving the structure of the tree and the sequents labeling the nodes.
\end{definition}

\begin{notation}
	We write $\vlderivation{\vlpr{\dD}{}{\vdash \Gamma}{}}$ for a proof $\dD$ with conclusion $\vdash \Gamma$, and $\vlderivation{\vlde{\dD}{}{\vdash \Gamma}{\vlhy{\vdash \Gamma'}}}$ for a derivation $\dD$ with conclusion $\vdash \Gamma$ and one open premise $\vdash \Gamma'$, or $\vlsmash{\vliqf{\rrule}{}{\vdash\Gamma}{\vdash \Gamma'}}$ if it is made only of instances of a same rule $\rrule$.
	We may denote by
	$\vlderivation{\vlin{\rrule}{}{\vdash \Gamma}{\multiprem{\vdash \Gamma_i}{i \in I}}}$ a rule with conclusion $\vdash \Gamma$ and a premise $\vdash \Gamma_i$ for each $i\in I$.
\end{notation}

The rules for the sequent calculus $\PIL$ on the formulas defined in \Cref{eq:formulas} are  given in \Cref{fig:rules}.
They include the standard rules for first-order multiplicative-additive linear logic (with a very limited set of atoms), plus a rule for the non-commutative connective $\lprec$.
It is a generalization of the standard rules for the tensor to the case where we have multiple pairs of formulas connected by $\lprec$.
In particular, the case $n=0$ corresponds to the rule for $\mixr$ \cite{retore:mix}.

\begin{figure}[t]
	\centering
	\adjustbox{max width=\textwidth}{$\begin{array}{c}
		\vlinf{\axrule}{}{\vdash {\lsend{\chan pq}{\xi}}, {\lrecv{\chan qp}{\xi}}}{}
	\qquad
		\vlinf{\lpar}{}{\vdash \Gamma, A\lpar B}{\vdash \Gamma, A, B}
	\qquad
		\vliinf{\ltens}{}{\vdash \Gamma, A\ltens B,\Delta}{\vdash \Gamma, A}{\vdash B, \Delta}
	\qquad
		\vliinf{\lprec}{\text{\scriptsize{$n\geq 0$}}}{
			\vdash \Gamma,\Delta, A_1\lprec B_1,\ldots, A_n\lprec B_n
		}{ \vdash \Gamma, A_1, \ldots, A_n}{\vdash \Delta, B_1, \ldots, B_n}
	\\[15pt]
		\vlinf{\lunit}{}{\vdash \lunit}{}
	\quad
		\vlinf{\lplus}{
			\text{\scriptsize{$k\in\intset1n$}}
		}{\vdash\Gamma, \bigoplus_{i=1}^n A_i}{ \vdash \Gamma, A_k}
	\quad
		\vliiinf{\lwith}{}{\vdash \Gamma,\bigwith_{i=1}^n A_i }{\vdash \Gamma,A_1}{\cdots}{\vdash \Gamma, A_n}
	\quad
		\vlinf{\forall}{\text{\scriptsize{$x\notin \freeof{\Gamma}$}}}{\vdash \Gamma, \lFa x{A}}{\vdash \Gamma, A}
	\quad
		\vlinf{\exists}{}{\vdash \Gamma, \lEx x{A}}{\vdash \Gamma, A\fsubst yx}
	\end{array}$}
	\caption{
		Sequent calculus rules for $\PIL$.
	}
	\label{fig:rules}
\end{figure}

We write $\proves[\PIL] \Gamma$ to denote that there is a derivation with conclusion a sequent $\Gamma$ of formulas using the rules in \Cref{fig:rules}.

We do not go into the details of the proof theory of $\PIL$, but the system satisfies both initial and final coherence (that is, the general axiom rule $\AXrule$ is derivable, and the cut rule $\cutr$ is admissible), end the structural congruence of processes is captured by the logical equivalence in the logic.

\subsection{Cyclic Proofs}\label{sec:pilu:cyclic}
To extend the expressiveness of the sequent calculus we consider in this paper, we allow for cyclic proofs, which are a generalization of standard (well-founded) proofs where the underlying proof tree can be infinite, but it is required to satisfy a certain correctness condition.

\begin{definition}
	A \defn{procedure definition} is an equation of the form $\genfProcDef$, where $\fX$ is a propositional variable (\defn{defined} by the procedure definition) and $A_{\fX}$ is a formula, called the \defn{body} of the procedure definition.
	The sequent system $\uPIL$ is defined by all rules in \Cref{fig:rules} plus the following \defn{unfolding} rule, replacing propositional variable $\fX$ in the sequent with a formula $A_{\fX}$ according to a procedural definition from a set $\fprocset$, used as parameter of the rule (and the system).
	\begin{equation}\label{eq:unfolding}
		\vlinf{\unfoldr}{
			\text{\scriptsize$\genfProcDef \in \fprocset$}
		}{\vdash \Gamma, \fX}{\vdash \Gamma, F_{\fX}}
	\end{equation}
	We denote by $\iPIL=\PIL\cup\set\unfoldr$.
\end{definition}

\begin{definition}
	Let $\fprocset$ be a set of procedure definitions, and $\Gamma$ a sequent.
	The set $\pnof\Gamma$ is defined as the set of process names occurring in the channel variables of the formulas in $\Gamma$, together with the process names occurring in the bodies of the procedure definitions in $\fprocset$ that are used to recursively unfold propositional variables occurring in $\Gamma$.
\end{definition}

\begin{remark}\label{rem:mumall}
	As mentioned at the beginning of this section, we could have considered a logical framework obtained by extending the system for $\mu\MALL$ with the non-commutative connective $\lprec$ and the first-order quantifiers ($\forall$ and $\exists$).
	In this setting, the encoding of processes as formulas would be similar to the one we have defined in \Cref{eq:fof}, but would have to use greatest fixpoint formulas.
	That is, the encoding of the atomic network $\pdo\pP$ should be the formula $\fX_1\ldots\fX_n. \fof\pP$, where $\fX_1,\ldots,\fX_n$ are the propositional variables corresponding to the procedure calls in $\pP$.
	
	In the following, the reader familiar with non-wellfounded proofs may find useful to keep in mind the nature of procedural definitions as greatest fixpoints to get an intuition about why the unfolding rules are used to define the progress condition.
	This ensures that the system is non-trivial, i.e., it cannot derive every formula of the language.
\end{remark}

As standard in the literature of non-wellfounded proofs, we define the notion of progress allowing us to identify infinite derivations that are sound, and therefore can be considered as proofs.

\begin{definition}
	An occurrence of a formula in the conclusion \resp{in a premise} of a rule, which does not appear in its premises \resp{not in its conclusion} is called \defn{principal} \resp{\defn{active}}.
	A \defn{\thread} in a derivation is a maximal sequence of formula occurrences $A_0, A_1, A_2, \ldots$ traced along the sequents of the tree, such that:
	\begin{itemize}
		\item if $A_i$ is principal, then $A_{i+1}$ is an active sub-formula from one of its premises.
		For instance, if $A_i=C\with D$, then $A_{i+1}\in\set{C,D}$, while if $A_i=\fX$ is an occurrence of a propositional variable, then $A_{i+1}$ is the formula $A_\fX$ from the procedural definition $\genfProcDef$ used in the unfolding rule that makes $A_i$ principal.
		
		\item if $A_i$ is not principal, then $A_{i+1}$ is the corresponding occurrence carried into a premise. 
		In the case of $\lwith$, the occurrence is duplicated in both premises, and the \thread\ chooses one copy.
	\end{itemize}
	
	A \defn{progress point} for a \thread in a derivation in $\uPIL$ is an occurrence of a propositional variable $\fX$ that occurs principal for the rule $\ldefr$ in the \thread.
\end{definition}

In the definition of \redtree, we identified the counterpart of the notions of fairness and progress from the standard literature on process calculi.
We now define the corresponding notions in the proof-theoretical setting.

\begin{definition}\label{def:fairProofs}
	A \defn{branch} of a derivation $\dD$ is a maximal path of sequents from the conclusion (the root) in the underlying tree of $\dD$.
	It is \defn{\fair} if:
	\begin{enumerate}
		\item\label{fair:1} every non-principal formula in a sequent $\Gamma$ of the branch is principal in a sequent above $\Gamma$ in the same branch;
		\item\label{fair:2} every time a formula of the form $\biglbra[\lab\in\aLset] A_{\lab}$ occurs principal infinitely often in the branch, each formula $A_{\lab}$ is also active infinitely often for every $\lab\in\aLset$.
	\end{enumerate} 

	A \defn{progress point} for a thread in a derivation is an element of the thread that is principal for some specific rules\footnote{
		For the case of $\mu\MALL$, progress points are defined by the $\nu$ rule.
		In the case of $\iPIL$, we progress points are defined by the $\unfoldr$ rule.
		We do define these latter because the proof theoretical framework we use, while based on $\iPIL$, has a different definition of progress points, using a rule we introduce later. See \Cref{def:chorl}.
	}.
	A derivation is:
	\begin{itemize}
		\item \defn{\progressing} if every infinite branch contains a \thread with infinitely many progress points.
		\item \defn{\regular} if it contains only a finite number of distinct non-isomorphic subderivations. It is \defn{\cyclic} if infinite and \regular. 
		\item \defn{\fair} if every \fair branch is \fair.
	\end{itemize} 

	If $\XS$ is a set of derivations, we say that a sequent $\Gamma$ is \defn{provable} in $\XS$, denoted $\proves[\XS] \Gamma$, if there is a derivation in $\XS$ with conclusion $\Gamma$.
\end{definition}

\begin{remark}\label{rem:fairImplyProgressing}
	Every \fair proof derivation is necessarily \progressing.
\end{remark}

\begin{remark}
	In process calculi, the \emph{justness} condition is used to rule out infinite executions where some enabled action is never performed.
	This, in proof-theoretical terms, corresponds Condition~\ref{fair:1} in \Cref{def:fairProofs}, which requires that every non-principal formula in a sequent $\Gamma$ of a branch is principal in a sequent above $\Gamma$ in the same branch.

	However, this condition does not have specific requirement on the behavior of a process when a choice is enabled infinitely often: it only requires that if a choice is enabled infinitely often, then it is eventually chosen.
	A stronger condition, called \emph{fairness}, is often used in process calculi to rule out infinite executions where a choice is always resolved in the same way.
	Condition~\ref{fair:2} in \Cref{def:fairProofs} is the proof-theoretical counterpart of the additional requirement imposed by fairness on the behavior of a process when a choice is enabled infinitely often, in addition to the requirement imposed by justness.
\end{remark}

\begin{notation}
	Regular derivations (also called \defn{cyclic} derivations) are a standard notion in the theory of cyclic proofs.
	They can be represented as finite graphs, where we mark with matching symbols rules whose conclusions are the roots of isomorphic subderivations.
	For example, the following infinite derivation on the left in $\iPIL$ with $\fprocset=\set{\procDef\fX{a \lprec \fX} \;,\; \procDef\fY{\lneg a \lprec \fY}}$, where the right-most branch continue infinitely reproducing the same structure,
	can be represented as show on the right, where we mark with $\star$ the two instances of $\lprec$ whose conclusions are the roots of isomorphic subderivations.
	\[
	\vlderivation{
		\vliin{\lprec}{}{
			\vdash 
			a \lprec \fX , 
			\lneg a \lprec \fY,A
		}{
			\vlin{\axrule}{}{\vdash a, \lneg a}{\vlhy{}}
		}{
			\vliq{\unfoldr}{}{
				\vdash \fX,\fY,A
			}{
				\vliin{\lprec}{}{
					\vdash 
					a \lprec \fX , 
					\lneg a \lprec \fY,A
				}{
					\vlin{\axrule}{}{\vdash a, \lneg a}{\vlhy{}}
				}{
					\vliq{\unfoldr}{}{
						\vdash \fX,\fY,A
					}{
						\vlin{\lprec}{}{
							a \lprec \fX , 
							\lneg a \lprec \fY,A
						}{\vlhy{\quad\vdots\quad}}
					}
				}
			}
		}
	}
	\quad\overset{\mathsf{notation}}=\quad
	\vlderivation{
		\vliin{\lprec}{\star}{
			\vdash 
			a \lprec \fX , 
			\lneg a \lprec \fY, A
		}{
			\vlin{\axrule}{}{\vdash a, \lneg a}{\vlhy{}}
		}{
			\vliq{\unfoldr}{}{
				\vdash \fX,\fY,A
			}{
				\vlin{\lprec}{\star}{
					\vdash 
					a \lprec \fX , 
					\lneg a \lprec \fY,A
				}{\vlhy{}}
			}
		}
	}
	\]

	If we assume the rule $\unfoldr$ to be the rule defining progress points, then the above derivation is \progressing, but it is not \fair because the right-most branch contains a \thread whose formulas (the $A$'s) are never active in any rule.
\end{notation}

\section{Proofs as Execution Trees}\label{sec:proofsAsExecutions}

In this section, we show how we can establish a correspondence between \redtrees of a network and a special class of proofs in $\iPIL$.
To simplify the correspondence, we introduce a set of derivable rules in $\iPIL$, called $\ChorL$, which are sound and complete with respect to set of networks-as-sequents derivable in $\iPIL$.
Then, we prove that the set of cyclic proofs in $\ChorL$ is sound and complete with respect to sequents that are encoding of \LF networks.

\subsection{The system $\ChorL$}\label{sec:proofsAsExecutions:chorl}

The sequent system $\ChorL$ is given by the rules in \Cref{fig:tchorl}, which are synthetic rules (in the sense of \cite{andreoli1992logic,miller:book}) derivable in  $\uPIL$.

\begin{figure}[t]
	\[\begin{array}{c}
		\vlinf{\Ccomr}{}{ 
			\vdash 
			\lsend{\chan pq}{e} \lprec F_\pp 
			\;,\;
			\lEx k\lrecv{\chan qp}{k} \lprec F_\pq 
			\;,\;
			\Gamma
		}{ 
			\vdash 
			F_\pp, 
			F_\pq\fsubst ek,
			\Gamma 
		}
		\qquad	
		\vliinf{\Cmixr}{
			\text{\scriptsize$\pnof{\Gamma_1}\cap\pnof{\Gamma_2}=\emptyset$}
		}{ \vdash \Gamma}{ \vdash \Gamma_1}{ \vdash \Gamma_2}
	\\[12pt]
		\vlinf{\Cwr}{
			\text{\scriptsize$\Gamma\neq\emptyset$}
		}{ 
			\vdash \Gamma, \lunit
		}{ \vdash \Gamma }
	\qquad
		\vlinf{\ldefr}{
			\text{\scriptsize$\genfProcDef \in \fprocset$}
		}{
			\vdash \Gamma, \fX
		}{
			\vdash \Gamma, F_\fX
		}
	\qquad
		\vlinf{\Crec}{
			\text{\scriptsize{\begin{tabular}{l}
				$\procDef{{X_i}}{{F_{\fX_i}}} \in \fprocset$ for all $i\in\intset1n$
				\\
				no subformula in $\Gamma$ is a $\fX\in\PXset$
			\end{tabular}}}
		}{
			\vdash \Gamma, \fX_1, \ldots, \fX_n
		}{
			\vdash \Gamma, F_{\fX_1}, \ldots, F_{\fX_n}
		}
	\\[12pt]
		\vlinf{\Caxr}{}{ 
			\vdash \lunit
		}{}
	\qquad
		\vlderivation{\vlin{\Cchor}{}{
			\vdash 
			\biglbra[{\lab\in\aLset}]
			\lsend{\chan pq}{\lab} \lprec F_\pp^\lab
			\;,\;
			\Gamma
		}{
			\multiprem{
				\vdash  F_\pp^\lab, \Gamma
			}{\lab \in \aLset} 
		}}
	\qquad
		\vlderivation{\vlin{\Clabr}{
			\text{\scriptsize{$\lab_i\in\aLset$}}
		}{
			\vdash 
			\biglsel[\lab\in\aLset]\lrecv{\chan pq}{\lab} \lprec F_\pp^\lab
			\;,\;
			\lsend{\chan qp}{\lab_i} \lprec F_\pq
			\;,\;
			\Gamma
		}{
			\vlhy{\vdash F_\pp^{\lab_i}, F_\pq, \Gamma}
		}}
	\end{array}\]

	\caption{Sequent calculus rules for the system $\ChorL$.}
	\label{fig:tchorl}
\end{figure}

\begin{proposition}\label{prop:chorl-sound}
	Rules in $\ChorL$ are derivable in $\uPIL$, that is, for every rule $\rrule$ in $\ChorL$, with conclusion $\Gamma$ and premises $\vdash\Gamma_1,\ldots,\vdash\Gamma_n$, there is a derivation $\dD$ in $\uPIL$ with conclusion $\vdash \Gamma$ and open premises $\vdash \Gamma_1,\ldots,\vdash \Gamma_n$.
\end{proposition}
\begin{proof}
	The rule $\Caxr$ is an instance of the $\lunit$-rule.
	The rule $\Cmixr$ is an instance of the $\lprec$-rule with $n=0$ with additional side condition on the process names of the sequents.
	The rule $\Cwr$ is derivable using $\Cmixr$ and $\Caxr$.
	The rule $\ldefr$ is a rule in $\uPIL$, and the rule $\Crec$ can be derived using a some $\ldefr$-rules.
	The rule $\Ccomr$ is derivable by using bottom-up, a $\exists$-rule to substitute the variable $k$ with the expression $e$ in $\fof{\precv{p}{k} \pseq \pQ}=\lEx k \lrecv pk \lprec \fof\pQ$, a $\lprec$-rule with $n=2$ with left premise a sequent of the form $\lsend{\chan pq}{e},\lrecv{\chan qp}{e}$, on which we can apply an $\axrule$-rule.
	The rule $\Cchor$ is derivable by using bottom-up, possibly multiple $\lwith$-rules branching the derivation into $|\aLset|$ branches, each starting with a $\lprec$-rule with $n=2$ with the left premise a sequent of the form $\lsend{\chan pq}{\ell},\lrecv{\chan qp}{\ell}$, on which we can apply an $\axrule$-rule.
\end{proof}

We define the class of proofs in $\ChorL$ which we are interested in, which are the ones that correspond to \redtrees of networks.
For this, the notion of progress point is defined in terms of the rule $\Crec$.
This condition ensures two important properties of the progressing proofs in $\gChorL$ by construction: they are \regular and \fair.
\begin{definition}\label{def:chorl}
	Let $\fprocset$ be a set of procedural definitions.
	A \defn{progress point} of a \thread in a derivation in $\ChorL$ is a formula occurrence that is principal for the rule $\Crec$.
	A proof in $\ChorL$ is \defn{\progressing} if every infinite branch contains a \thread with infinitely many progress points;
	it is \defn{\nice} if it is \progressing 
	and such that 
	the premise of any rule other than $\Cmixr$ cannot be the premise of a $\Cmixr$-rule,
	and that the active formulas of $\Cchor$-rule is always principal for a $\Clabr$-rule.

	We define $\gChorL$ \resp{$\nChorL$} as the set of progressing \resp{\nice} proofs in $\ChorL$.
\end{definition}
\begin{lemma}\label{cor:chorl_cyc}
	Let $\nN$ be a \net.
	If $\proves[\gChorL]\fof{\nN}$, then there is a \regular proof of $\fof{\nN}$ in $\gChorL$.
\end{lemma}
\begin{proof}
	Let $\dD$ be a derivation in $\gChorL$ with conclusion
    $\fof{\nN}$. There are only finitely many sequents that can occur in
    $\dD$ because $\fof{\nN}$ contains only finitely many subformulas, 
	$\fprocset$
    is finite and closed,
	and, except for $\unfoldr$ and $\Crec$, rules in $\ChorL$ do not introduce new formulas. 
	We conclude observing that, since any infinite branch of $\dD$ must repeat a sequent,
	$\dD$ is a \regular proof.
\end{proof}


\begin{lemma}\label{cor:chorl_ok}
	Let $\nN$ be a \net.
	If $\proves[\gChorL]\fof{\nN}$, then $\proves[\nChorL]\fof{\nN}$.
\end{lemma}
\begin{proof}
    Let $\dD$ be a proof in $\gChorL$ of $\fof{\nN}$. We transform $\dD$ into a \nice proof in two steps.
    First, if $\nN = \nN[1] \ppar \nN[2]$ with disjoint supports, 
	then we can insert a $\Cmixr$-rule to the sequent $\fof{\nN}=\fof{\nN[1]}, \fof{\nN[2]}$, and split the subproof above the sequent $\nN$ into the two proofs in the premises of the $\Cmixr$-rule, one with conclusion $\fof{\nN[1]}$ and the other with conclusion $\fof{\nN[2]}$, made of the rules applied to the formulas in $\fof{\nN[1]}$ and $\fof{\nN[2]}$, respectively.
	Then, we permute the $\Clabr$-rules downward until they reach the $\Cchor$-rules that make their active formulas principal, using the permutations of \Cref{fig:equivTCHORL}. The resulting derivation is \nice by construction.
\end{proof}

\def\backedge{{\tiny \rclr{\hookrightarrow }}}
\begin{figure}[t]
	\centering
	\adjustbox{max width=\textwidth}{$
	\begin{array}{c}
		\vlderivation{
			\vlin{\rclr{\rrule[1]}}{}{\vdash \Gamma, \rclr{\Delta}, \Theta }{
				\vlin{\bclr{\rrule[2]}}{}{ 
				\vdash \Gamma, \rclr{\Delta_1}, \bclr{\Theta}}{\vlpr{\dD}{}{\vdash \Gamma, \Delta_1, \bclr{\Theta_1}}}
			}
		}
	\peq
		\vlderivation{
			\vlin{\bclr{\rrule[2]}}{}{\vdash \Gamma, \Delta, \bclr{\Theta}}{
				\vlin{\rclr{\rrule[1]}}{}{\vdash \Gamma,\rclr{\Delta},\bclr{\Theta_1}}{\vlpr{\dD}{}{\vdash \Gamma, \rclr{\Delta_1}, {\Theta_1}}}
			}
		}
	\qquad
		\vlderivation{
			\vlin{\rclr{\Cchor}}{}{\vdash \Gamma, \rclr{\Delta}, \Theta}{
				\multiprem{
					\vlin{\bclr{\Cchor}}{}{
						\vdash \Gamma, \rclr{\Delta_{\lab}}, \bclr{\Theta}
					}{\multiprem{
						\vlpr{\dD_{\lab,\lab'}}{}{
							\vdash \Gamma, \rclr{\Delta_{\lab}}, \bclr{\Theta_{\lab'}}
						}
					}{\lab' \in \aLset'}
					}
				}{\lab \in \aLset}
			}
		}
	\peq
		\vlderivation{
			\vlin{\bclr{\Cchor}}{}{\vdash \Gamma, \Delta, \bclr{\Theta}}{
				\multiprem{
					\vlin{\rclr{\Cchor}}{}{
						\vdash \Gamma, \rclr{\Delta}, \bclr{\Theta_{\lab'}}
					}{\multiprem{
						\vlpr{\dD_{\lab_1,\lab'}}{}{
							\vdash \Gamma, \rclr{\Delta_{\lab_1}}, \bclr{\Theta_{\lab'}}
						}
					}{\lab\in \aLset}
					}
				}{\lab'\in \aLset'}
			}
		}
	\\[12pt]
		\vlderivation{
			\vlin{\rclr{\rrule[1]}}{}{\vdash \Gamma, \rclr{\Delta}, \Theta}{
				\vlin{\bclr{\Cchor}}{}{ 
				\vdash \Gamma, \rclr{\Delta_1}, \bclr{\Theta}
				}{
					\multiprem{\vlpr{\dD_{\lab}}{}{
						\vdash \Gamma, \Delta_1, \bclr{\Theta_{\lab}}
					}}{\lab \in \Lset}
				}
			}
		}
	\peq
		\vlderivation{
			\vlin{\bclr{\Cchor}}{}{
				\vdash \Gamma, \Delta, \bclr{\Theta}
				}
			{
				\multiprem{\vlin{\rclr{\rrule[1]}}{}{
					\vdash \Gamma, \rclr{\Delta}, \bclr{\Theta_{\lab}}
				}{\vlpr{\dD_{\lab}}{}{\vdash \Gamma, \rclr{\Delta_1}, {\Theta_{\lab}}}}}{\lab \in \Lset}
			}
		}
	\qquad
		\vlderivation{\vlin{\Crec}{}{
			\vdash \Gamma, \fX_1, \ldots, \fX_n
		}{
			\vlpr{\dD}{}{\vdash \Gamma, F_{\fX_1}, \ldots, F_{\fX_n}}
		}}
	\peq
		\vlderivation{
			\vlin{\ldefr}{}{
				\vdash \Gamma, \fX_1, \ldots, \fX_n
				}{
					\vldd{}{}{\vdash \Gamma, F_{\fX_1}, \ldots, \fX_n}{
						\vlin{\ldefr}{}{
							\Gamma, F_{\fX_1}, \ldots, F_{\fX_{n-1}},  \fX_n
						}{
							\vlpr{\dD}{}{\vdash \Gamma, F_{\fX_1}, \ldots, F_{\fX_n}}
						}
						}
				}
		}
	\\[12pt]
		\vlderivation{
			\vliin{\Cmixr}{}{\vdash \Gamma_1, \Gamma_2, \Delta}{
				\vlin{\bclr{\rrule[3]}}{}{  
				\vdash \Gamma_1, \bclr{\Delta}
				}{\vlpr{\dD_{1}}{}{\vdash \Gamma_1, \bclr{\Delta_1}}}
			}{
				{\vlpr{\dD_2}{}{\vdash \Gamma_2}}
			}
		}
	\peq
		\vlderivation{
			\vlin{\bclr{\rrule[3]}}{}{\vdash \Gamma_1, \Gamma_2, \bclr{\Delta}}{
				\vliin{\Cmixr}{}{\vdash \Gamma_1, \Gamma_2, \bclr{\Delta_1}}{\vlpr{\dD_1}{}{\vdash \Gamma_1, \Delta_1}}{\vlpr{\dD_2}{}{\vdash \Gamma_2}}
			}
		}
	\qquad
		\vlderivation{
			\vlin{\rclr{\Cchor}}{}{\vdash \Gamma_1, \Gamma_2, \rclr{\Delta}}{
				\multiprem{\vliin{\Cmixr}{}{\vdash \Gamma_1, \Gamma_2, \rclr{\Delta_\lab}}{\vlpr{\dD_1}{}{\vdash \Gamma_1, \Delta_{\lab}}}{\vlpr{\dD_2}{}{\vdash \Gamma_2}}}{\lab \in \Lset}
			}
		}
	\peq
		\vlderivation{
			\vliin{\Cmixr}{}{\vdash \Gamma_1, \Gamma_2, \Delta}{
				\vlin{\rclr{\Cchor}}{}{  
				\vdash \Gamma_1, \rclr{\Delta}
				}{
					\multiprem{\vlpr{\dD_{\lab}}{}{\vdash \Gamma_1, \rclr{\Delta_{\lab}}}}{\lab \in \Lset}
				}
			}{
				{\vlpr{\dD_2}{}{\vdash \Gamma_2}}
			}
		}
	\end{array}
	$}

	\caption{
		Rule permutations defining the proof equivalence in $\ChorL$,
		with $\rrule[1], \rrule[2] \in \set{\Ccomr,\ldefr, \Crec}$
		and $\rrule[3]\in \set{\Ccomr,\ldefr}$.
	}
	\label{fig:equivTCHORL}
\end{figure}

\subsection{Proofs as Execution Trees, and Lock-freedom as provability}\label{sec:proofsAsExecutions:lockfreedom}

Each rule in $\ChorL$ is designed to faithfully correspond to a single transition in the network semantics.
In particular, we here show that the results in \cite{acc:man:mon:ESOP25} can be extended to networks with infinite behavior.


\begin{figure}[t]
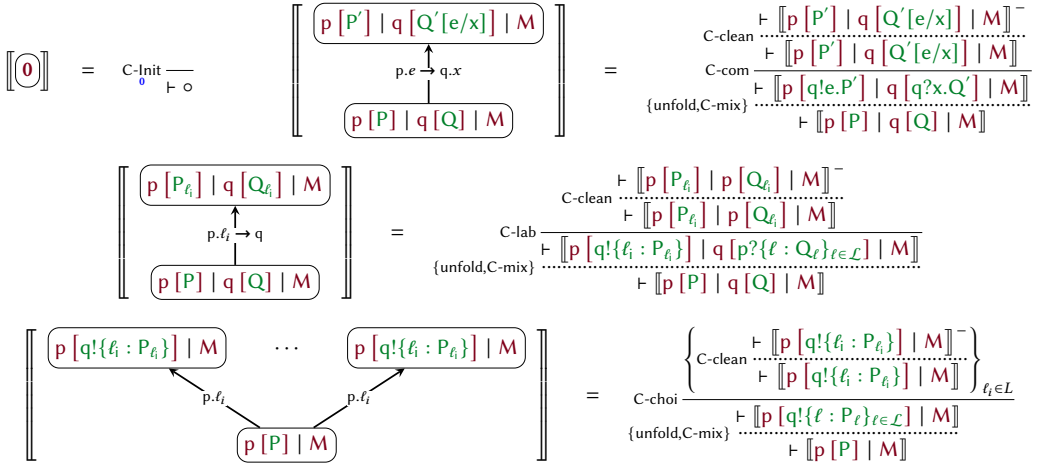

	\adjustbox{max width=\textwidth}{$\begin{array}{c}
		\fof{\vmod1{\nnil}}
	\quad=\quad
		\vlderivation{
			\vlin{\Caxr[\cnil]}{}{\vdash \lunit}{\vlhy{}}	
		}
	\qquad\qquad
		\fof{\begin{array}{c}
			\vmod2{ \pdo{\pPp} \ppar \pdo[q]{\pQp\psubst ex}  \ppar \nM}
			\\[30pt]
			\vmod1{ \pdo{\pP} \ppar \pdo[q]{\pQ}  \ppar \nM}
		\end{array}}
		\Dledges{mod1/mod2/{\tocom peqx}}
	\quad=\quad	
	\vlderivation{
		\vlid{\set{\unfoldr,\Cmixr}}{}{
			\vdash \fof{\pdo{\pP} \ppar \pdo[q]{\pQ}  \ppar \nM}
		}{
			\vlin{\Ccomr}{}{
				\vdash \fof{\pdo{\psend qe \pseq \pPp} \ppar \pdo[q]{\precv qx \pseq \pQp }\ppar \nM}
			}{
				\vlid{\Cwr}{}{
					\vdash \fof{\pdo{\pPp} \ppar \pdo[q]{\pQp\psubst ex} \ppar \nM}
				}{\vlhy{
					\vdash \fofclean{\pdo{\pPp} \ppar \pdo[q]{\pQp\psubst ex} \ppar \nM}
				}}
			}
		}
	}
	\\\\
		\fof{\begin{array}{c}
			\vmod2{
				\pdo{\pP[\lab_i]} \ppar \pdo[q]{\pQ[\lab_i]}\ppar \nM
			}
			\\[30pt]
			\vmod1{
				\pdo{\pP} \ppar \pdo[q]{\pQ}\ppar \nM
			}
		\end{array}}
		\Dledges{mod1/mod2/{\tolab pq}}
	\quad=\quad		
	\vlderivation{
		\vlid{\set{\unfoldr,\Cmixr}}{}{
			\vdash \fof{\pdo{\pP} \ppar \pdo[q]{\pQ}  \ppar \nM}
		}{
			\vlin{\Clabr}{}{
				\vdash \fof{\pdo{\pgenselectl} \ppar \pdo[q]{\pgenrecvl} \ppar  \nM}
			}{
				\vlid{\Cwr}{}{
					\vdash \fof{\pdo{\pP[\lab_i]} \ppar \pdo{\pQ[\lab_i]} \ppar  \nM}
				}{\vlhy{
					\vdash \fofclean{\pdo{\pP[\lab_i]} \ppar \pdo{\pQ[\lab_i]} \ppar  \nM}
				}}
			}
		}
	}
	\\\\
		\fof{\begin{array}{c}
			\vmod2{\pdo{\pgenselectl} \ppar \nM}
			\qquad
			\cdots
			\qquad
			\vmod3{\pdo{\pgenselectl} \ppar \nM}
			\\[30pt]
			\vmod1{\pdo{\pP} \ppar \nM}
		\end{array}}
		\Dledges{mod1/mod2/{\tochoi p},mod1/mod3/{\tochoi p}}
	\quad=\quad
		\vlderivation{
			\vlid{\set{\unfoldr,\Cmixr}}{}{
				\vdash \fof{\pdo{\pP} \ppar \nM}
			}{
				\vlin{\Cchor}{}{
					\vdash \fof{\pdo{\pgensendl} \ppar\nM}
				}{
					\multiprem{
						\vlid{\Cwr}{}{
							\vdash \fof{\pdo{\pgenselectl} \ppar \nM}
						}{\vlhy{
							\vdash \fofclean{\pdo{\pgenselectl} \ppar \nM}
						}}
					}{\lab_i \in L}
				}
			}
		}
	\end{array}
	$}
	\caption{
		Translation of execution trees to derivations, where dotted lines represents possible applications of $\ldefr$.
		Dotted lines in the derivations represent possible applications the rule noted on the left of the line.
	}
	\label{fig:ctTOder}
\end{figure}

\begin{lemma}\label{lem:proofsAsExecutions}
	Let $\nN$ be a \net.
	There is a one-to-one correspondence between \redtrees of $\nN$ and proofs of $\fof{\nN}$ in $\ChorL$.
\end{lemma}
\begin{proof}
    We extend the translation of networks into sequents $\fof{\cdot}$ to \redtrees as show in \Cref{fig:ctTOder} allowing us to translate a \redtree $\predtree\nN$ into a derivation in $\ChorL$.
\end{proof}
\begin{corollary}
	The translation $\fof\cdot$ maps \regular \resp{\ok} \redtrees into \regular \resp{\nice} proofs in $\ChorL$.
\end{corollary}
\begin{proof}
	For \ok \redtrees, every non-nil state eventually advances, so the translation produces a \fair proof (thus progressing ones).
\end{proof}

We can now establish the central result of this section: the system $\nChorL$ is a proof system deciding lock-freedom for networks.

\begin{theorem}\label{thm:LFinTCHORL}
	A \net $\nN$ is \LF if and only if $\proves[\gChorL] \fof \nN$.
\end{theorem}
\begin{proof}
	By \Cref{thm:LF}, if $\nN$ is \LF, then there is a \ok \redtree of $\nN$.
	By \Cref{lem:proofsAsExecutions}, each \redtree of $\nN$ can be translated into a proof in $\gChorL$ of $\fof{\nN}$. 
	Therefore, if $\nN$ is \LF, then $\proves[\gChorL] \fof{\nN}$.

	The implication from provability in $\gChorL$ to \LF relies crucially on
	the fact that our networks are race-free: each send action has a unique
	matching receive, and each $\Cchor$-rule has a unique set of matching
	$\Clabr$-rules. Hence, every rule in the proof corresponds unambiguously
	to a single network transition, and the progress condition of $\gChorL$
	ensures that every such transition is eventually taken.
\end{proof}

As a corollary, we obtain a result that serves as basis for the definition of the extraction procedure in \Cref{sec:proofsAsChoreos}, proving that if a network $\nN$ is \LF, then there is a proof of $\fof\nN$ in $\nChorL$.

\begin{corollary}\label{cor:lf_iif_proof}
	Let $\nN$ be a \net.
	If $\nN$ is \LF, then 
	$\fof \nN$ admits a proof in $\nChorL$.
\end{corollary}
\begin{proof}
	Consequence of \Cref{thm:LFinTCHORL} and \Cref{cor:chorl_ok}.
\end{proof}

\section{Proofs as Choreographies}\label{sec:proofsAsChoreos}

\def\Cchorr{\mathsf{limbo}}
\def\preChorL{\mathsf{pre-}\ChorL}
\def\concl#1{\mathrm{concl}(#1)}

\def\unfof#1{\fof{#1}^{-1}}
\def\proofsubst#1#2{[\left\{#1\right\} \to \left\{#2\right\}]}

In this section we establish a correspondence between choreographies and proofs in $\gChorL$.
We show that we can translate a choreography $\cC$ into a proof with conclusion a sequent describing a network $\nN$, such that $\nN$ is bisimilar to $\cC$.
Then, we show that we can reverse the argument, and define a choreography from a successful proof search in $\PIL$ of a sequent describing a network $\nN$, such that the extracted choreography is bisimilar to $\nN$.

\subsection{Endpoint Projection, Logically}\label{sec:proofsAsChoreos:epp}

To reconstruct the endpoint projection of a choreography $\cC$ as a proof translation, we proceed in two steps. 
We first define a system $\preChorL$ obtained by relaxing the rule $\Cchor$ of $\ChorL$ into a rule $\Cchorr$ whose premises may have completely unrelated contexts.
This allows us to define a total function $\DER{\cdot}$ translating any choreography into a derivation in $\preChorL$, without checking any condition equivalent to projectability. 
We then identify the projectable choreographies as those for which every occurrence of $\Cchorr$ in $\DER{\cC}$ can be replaced by the sound rule $\Cchor$, yielding a proof
$\DofC{\cC}$ in $\ChorL$.

To translate a recursive call $\cX$ into a proof, we need to identify the processes having genuinely infinite behavior in $\cX$.
These processes are the ones that receive a recursive propositional variable $\fX_{\pp}$ in the translation.
In fact, the choreography associated with $\cX$ may contain finite branches in which some processes are not necessarily involved in the recursion, as show in the example in \Cref{ex:intro:unprojectable:proof}.
We capture this distinction using \emph{recursive names}, defined as follows.
\begin{definition}\label{def:finnames}
    The function \defn{recursive names} $\recnof\cdot$ mapping choreographies to sets of process names that are involved in infinite behaviors, inductively defined as follows.
    \[
        \recnof{\cC} =
        \begin{cases}
			\emptyset 
		& 
			\text{if } \cC \text{ is finite}
		\\
            \recnof{\cC[1]} \cup \recnof{\cC[2]}
		& 
			\text{if } \cC = \cC[1] \ctens \cC[2]
        \\
            \recnof{\cC[i]}
		& 
			\text{if } \cC = \cX[i] \text{ and } \procDef{\cX[i]}{\cC[i]}
		\\
            \set{\pp,\pq} \cup \recnof{\cC[1]}
		& 
			\text{if } \cC = \cgencom \seq \cC[1]
			\text{ (and  $\cC[1]$ is not finite)}
		\\
            \set{\pp,\pq} \cup \bigcup_{\lab\in\aLset}\recnof{\cC[\lab]}
		& 
			\text{if } \cC = \cgenchois
			\\ &
			\text{ (and $\cC[\lab]$ is not finite for some  $\lab\in\aLset$)}
        \end{cases}
    \]
\end{definition}

\begin{example}\label{ex:rnof}
	Consider the choreography $\cC= \cchoi p{\set{L,R}}q{
	\begin{array}{ll}
		L :& \com\pp e \pr x \seq \cnil
	\\
		R :& \cX
	\end{array}
	}
	$
	with 
	$\gencProcDef$ from \Cref{ex:intro:unprojectable:proof}.
	We have $\recnof{\cX} = \set{\pp, \pq}$, since $\pr$ occurs only in the finite branch $\ccom\pp e \pr x \seq \cnil$ of the continuation, and is therefore not included in $\recnof{\cC}$.
\end{example}

We can now define a first map $\DER{\cdot}$ from choreographies to derivations in a system $\preChorL$ obtained by relaxing the rule $\Cchor$ of $\ChorL$ into a rule $\Cchorr$ whose premises may have unrelated contexts.%
\footnote{
	The reader familiar with the literature on choreographies may interpret the rule $\Cchorr$ as a way to force projection even when merge is not defined, by allowing all processes, not only the two involved in the \emph{choice} primitive, to keep the possibility of making an internal choice to adapt to the eventual choice made by the other processes.
}

\begin{definition}
	We define the system $\preChorL=\ChorL\cup\set{\Cchorr}$, where the rule $\Cchorr$ is defined as follows.
	\begin{equation}\label{eq:relaxCchor}
		\vlderivation{\vlin{\Cchorr}{}{
			\vdash 
			\biglbra[{\lab\in\aLset}]
				\lsend{\chan pq}{\lab} \lprec F_\pp^\lab
			\;,\;
			\biglsel[\lab\in\aLset]
				\lrecv{\chan qp}{\lab} \lprec F_\pq^\lab
			\;,\;
			\biglsel[\lab\in\aLset]
				F_{\pr_1}^\lab
			,\ldots,
			\biglsel[\lab\in\aLset]
				F_{\pr_k}^\lab
		}{
			\multiprem{
				\vdash F_\pp^\lab, F_\pq^\lab, F_{\pr_1}^\lab,\ldots, F_{\pr_k}^\lab
			}{\lab \in \aLset} 
		}}
	\end{equation}

	We associate a derivation $\DER{\cC}$ in $\preChorL$ to a choreography $\cC$ by induction on the structure of $\cC$, using the translation $\DER\cdot$ defined in \Cref{fig:choreoTrans}.
	We say that $\cC$ is \defn{projectable} if every occurrence of 
	$\Cchorr$ in $\DER{\cC}$ is such that 
	we can replace each non-principal formula $\biglsel[\lab\in\aLset]F_{\pr_i}^\lab$ with a single formula $F_{\pr_i}$, for each $i\in\intset1k$, and the resulting derivation is still a proof in $\ChorL$ (see \Cref{ex:extraction}).
	
	If $\cC$ is projectable, then we define $\DofC{\cC}$ as the proof in $\ChorL$ obtained by translation $\DofC\cdot$ defined as in \Cref{fig:choreoTrans} by changing the way we translate the $\Cchorr$-rules, which are translated as shown in the bottom of \Cref{fig:choreoTrans} with respect to the translation $\DER\cdot$.
\end{definition}

\begin{remark}
	We define the translation of a choreography into a derivation in $\preChorL$ by progressively decomposing the choreography. To each choreography $\cC$ we associate a sequent 
	$\Gamma_{\cC}=F_{\pp_1},\ldots,F_{\pp_n}$, where $\pnof{\cC}=\set{\pp_1,\ldots,\pp_n}$.
	Each clause of \Cref{fig:choreoTrans} extends the derivation associated with the continuation of the choreography by adding the rule corresponding to its outermost constructor. 
	The formulas introduced by the rule replace the formulas associated with the involved processes throughout the open derivation below it. 
	In particular, in the communication case, the formulas introduced by $\Ccomr$ replace the formula
	$F_\pp$ \resp{$F_\pq$} with $\lsend{\chan pq}{e}\lprec F_\pp$ \resp{$\lrecv{\chan pq}{k}\lprec F_\pq$} in the derivation below.
	In the recursive case, only the processes in $\recnof{\cX} \cap \apnof{\cC}$, with $\cX \triangleq \cC$, are associated with propositional variables. 
	Thus, procedure calls are introduced only for the processes occurring in the recursive behavior of $\cX$; 
	processes occurring solely in finite sub-choreographies are instead represented by their finite formulas.
\end{remark}

\begin{example}\label{ex:extraction}
	Consider the choreography $\cC$ from \Cref{ex:intro:unprojectable:proof}.
	When we translate $\cC$ into a derivation in $\DER{\cC}$, we start by translating the outermost constructor of $\cC$ into a $\Cchorr$-rule, obtaining the following derivation in $\preChorL$
	$$
	\begin{array}{c}
	\DER{\cC}=
		\vlderivation{
			\vliin{\Cchorr}{}{
				\vdash 
				\left(
					\lsend{\chan{p}{q}}{L}
					\lprec
					F_\pp^L
				\right)
				\lwith
				\left(
					\lsend{\chan{p}{q}}{R}
					\lprec
					F_\pp^R
				\right)
				,
				\left(
					\lrecv{\chan{q}{p}}{L}
					\lprec
					F_\pq^L
				\right)
				\lplus
				\left(
					\lrecv{\chan{q}{p}}{R}
					\lprec
					F_\pq^R
				\right)
				,
				\left(
					F_\pr^L
				\right)
				\lplus
				\left(
					F_\pr^R
				\right)
			}{
				\vlpr{
					\DER{\ccom perx \seq \cnil}
				}{}{
					\vdash
					\lsend{\chan{p}{r}}{e},
					\lunit,
					\lrecv{\chan{r}{p}}{x}
				}
			}{
				\vlpr{
					\DER{\cX}
				}{}{
					\vdash 
					\fX_{\pp},
					\fX_{\pq},
					F_\pr^R
				}
			}
		}
	\end{array}
	$$
	which would require to satisfy the following equations:
	$$
	F_\pp^L=\lsend{\chan{p}{r}}{e}\lprec \lunit
	\qquad
	F_\pq^L=\lunit
	\qquad
	F_\pr^L=\lrecv{\chan{r}{p}}{x}\lprec \lunit
	\qquad
	F_\pp^R=\fX_{\pp}
	\qquad
	F_\pq^R=\fX_{\pq}
	$$
	
	In order to define $\DofC{\cC}$, it suffices to check if the equation $F_\pr^R=F_\pr^L$ admits a solution with the current constraints, which is the case here.
	If we assume $F_\pr^R=F_\pr^L=\lrecv{\chan{r}{p}}{x}\lprec \lunit$, we can safely continue translating the right branch obtaining the following derivation in $\ChorL$.
	$$
	\DofC{\cX}
	=
	\vlderivation{
		\vlin{\Crec}{}{
			\vdash 
			\fX_{\pp},
			\fX_{\pq},
			\lrecv{\chan{r}{p}}{x}\lprec \lunit
		}{
			\vlpr{
				\DofC{\cC}
			}{}{
				\vdash 
				\left(
					\lsend{\chan{p}{q}}{L}
					\lprec
					F_\pp^L
				\right)
				\lwith
				\left(
					\lsend{\chan{p}{q}}{R}
					\lprec
					F_\pp^R
				\right)
				,
				\left(
					\lrecv{\chan{q}{p}}{L}
					\lprec
					F_\pq^L
				\right)
				\lplus
				\left(
					\lrecv{\chan{q}{p}}{R}
					\lprec
					F_\pq^R
				\right)
				,
				\lrecv{\chan{r}{p}}{x}\lprec \lunit
			}
		}
	}
	$$
	The derivation $\DofC{\cC}$ is defined as in \Cref{ex:intro:unprojectable:proof}.
\end{example}

\begin{figure}[t]
	\def\myskip{\hskip2em}
	\def\Trans#1{}
 \centering
	\adjustbox{max width=\textwidth}{$
	\begin{array}{c}
		\TofD{\cnil} = \vlinf{\init}{}{
			\vdash \lunit
		}{}
	\myskip
		\TofD{\cX}
		=
		\vlderivation{
			\vlin{\Crec}{
				\text{\scriptsize{
					$\procDef{{\fX_{\pp_i}}}{{F_{\pp_i}}} \in\fprocset$
				}}
			}{
				\vdash \fX_{\pp_1},\dots,\fX_{\pp_n},\Gamma
			}{
				\vlin{}{}{
					\vdash F_{\pp_1},\dots,F_{\pp_n},\Gamma
				}{
					\vlpr{\TofD{\cC}}{}{
						\vdash \Gamma_{\cC}
					}
				}
			}
		}	
		\qwith 
		\begin{cases}
			\procDef{\cX}{{\cC}} \in \cprocset
			\\
			\text{and }
			\\
			\set{\pp_1,\dots,\pp_n}=\recnof{\cX}\cap\apnof\cC
		\end{cases}
	\\\\
		\TofD{\cC[1]\ctens\cC[2]}
		=
		\vlderivation{
			\vliin{\Cmixr}{}{
				\vdash \Gamma
			}{
				\vlpr{\TofD{\cC[1]}}{}{\vdash \Gamma_{\cC[1]}}
			}{
				\vlpr{\TofD{\cC[2]}}{}{\vdash \Gamma_{\cC[2]}}
			}
		}
	\myskip\myskip
		\TofD{\cgencom\seq\cC}
		=
		\vlderivation{
			\vlin{\Ccomr}{}{
				\vdash 
				\lsend{\chan pq}{e}\lprec F_\pp
				,
				\lrecv{\chan pq}{k}\lprec F_\pq
				,
				\Gamma
			}{
				\vlid{}{}{
					\vdash 
					F_\pp
					\;,\;
					F_\pq\fsubst ek
					\;,\;
					\Gamma
				}{
					\vlpr{\TofD{\cC\cvsubst ek}}{}{
						\vdash \Gamma_{\cC\cvsubst ek}
					}
				}
			}
		}
	\\\\\hline\\
		\DER{\cchoi p{\aLset}q{\lab : \cC[\lab] \mid \lab \in \aLset'}}
		=
		\vlderivation{
			\vlin{\Cchorr}{}{
				\vdash 
				\Biglbra[\lab\in\aLset]
				{
					\left(
					\lsend{\chan{p}{q}}{\lab}
					\lprec
					F_{\pid p}^{\lab}
					\right)
				}
				\;,\;
				\Biglsel[{\lab\in\aLset'}]
				{
					\left(
					\lrecv{\chan{q}{p}}{\lab}
					\lprec
					F_\pq^{\lab}
					\right)
				}
				\;,\;
				\Biglsel[\ell\in\aLset] F_{\pr_1}^\lab,\ldots, \Biglsel[\ell\in\aLset] F_{\pr_k}^\lab
			}{
				\multiprem{
					\vlderivation{\vlid{\Cwr}{}{
						\vdash 
						F_\pp^\lab, F_\pq^{\lab}
						\;,\;
						F_{\pr_1}^\lab
						\;,\ldots,\;
						F_{\pr_k}^\lab
					}{
						\vlpr{\DER{\cC[\lab]}}{}{
							\vdash 
							\Gamma_{\cC[\lab]}
						}
					}}
				}{\lab\in\aLset}
			}
		}
	\\\\\hline\\
		\DofC{\cchoi p{\aLset}q{\lab : \cC[\lab] \mid \lab \in \aLset'}}
		=
		\vlderivation{
			\vlin{\Cchor}{}{
				\vdash 
				\Biglbra[\lab\in\aLset]
				{
					\left(
					\lsend{\chan{p}{q}}{\lab}
					\lprec
					F_{\pid p}^{\lab}
					\right)
				}
				\;,\;
				\Biglsel[\lab\in\aLset']
				{
					\left(
					\lrecv{\chan{q}{p}}{\lab}
					\lprec
					F_{\pid q}^{\lab}
					\right)
				}
				\;,\;
				\Gamma
			}{
				\multiprem{
					\vlderivation{
						\vlin{
							\Clabr
						}{
							\text{\scriptsize{$\lab_i\in\aLset$}}
						}{
							\vdash 
							\lsend{\chan pq}{\lab_i} \lprec F_\pp^{\lab_i}
							\;,\;
							\biglsel[\lab\in\aLset']\lrecv{\chan qp}{\lab} \lprec F_\pq^\lab
							\;,\;
							\Gamma
						}{
							\vlhy{
								\vdash 
								F_\pp^{\lab_i} 
								\;,\; 
								F_\pq^{\lab_i}
								\;,\;
								\Gamma
							}
						}
					}
				}{\lab_i\in\aLset}
			}
		}
	\end{array}
	$}	
	\caption{
		Translation of a choreography into a pre-derivation, with $\TofD{\cdot}\in\set{\DER{\cdot},\DofC{\cdot}}$ and $\Gamma_{\cC} = F_{\pp_1} , \dots, F_{\pp_n}$ with $\pnof{\cC} = \set{\pp_1, \dots, \pp_n}$.
		The dotted lines represent the possible applications instances of the $\Cwr$.
	}
	\label{fig:choreoTrans}
\end{figure}

\begin{remark}\label{ex:noRuntimeTerms}
	The way the choice is translated into a proof, allows us to have a better handling of the knowledge of choice in the choreography as shown in the previous example.
	
	In fact, the choreography in \Cref{ex:intro:unprojectable:choreo} cannot be projected using standard techniques (as, e.g., in \cite{montesi:book}) because of the lack of knowledge of choice for $\pr$.
	However, the conclusion sequent of $\DofC{\cC}$ is bisimilar to the following network we obtain using our new endpoint projection technique.
	$$
	\pdo[p]{\matpsendl{q}{L: \psend re, R: {\pX[\pp]}}}
	\ppar
	\pdo[q]{\matprecvl{q}{L: \pnil, R: {\pX[\pq]}} }
	\ppar
	\pdo[r]{\precv px} 
	$$

	In fact, the \fairness property of the proof $\DofC{\cC}$ highlight the fact that because $\pr$ is not involved in the recursive behavior of $\cC$, and $\pr$ does not need to be aware of the choice made by $\pp$ and $\pq$.
\end{remark}

\begin{remark}\label{rmk:byconstruction}
	Notice that by construction the conclusion of $\DofC\cC$ is always of the form $\Gamma = \fof N$.
\end{remark}

We can now show that if two choreographies are projectable, and their translations are equivalent proofs in $\gChorL$ modulo the equivalence relation $\peq$ over derivations in $\ChorL$ inductively generated by the rule permutations in \Cref{fig:equivTCHORL}, then the two choreographies are bisimilar.

\begin{lemma}\label{lem:proof_equiv_chor_bisim}
	Let $\cC[1]$ and $\cC[2]$ projectable choreographies.
	If $\DofC{\cC[1]} \peq \DofC{\cC[2]}$, then $\cC[1] \bisim \cC[2]$.
\end{lemma}
\begin{proof}
	Since we are assuming $\DofC{\cC[1]} \peq \DofC{\cC[2]}$ if there is a finite sequence of rule permutations transforming $\DofC{\cC[1]}$ into $\DofC{\cC[2]}$, we can reason by induction on the number of rule permutations in the sequence.
	In fact, each rule permutation requires that the process names involved in the active formulas of the two rules are disjoint, and therefore preserves the bisimilarity of the corresponding choreographies via delay rules.
\end{proof}

We can now prove the main result of this section, which is a reformulation of the endpoint projection of choreographies in logical terms. For this purpose, we recall the definition of bisimulation between choreographies and networks.
\begin{definition}\label{def:wbisim}
	A binary relation $\mathcal{R}$ between \chors and \nets is a \defn{bisimulation of \nets and \chors} if for all pair $(\cC,\nN)\in\mathcal{R}$,
	for all 
	$\mu \in \set{\tocom peqx , \ \tochoi[\lab] p, \ \tolab[\lab] pq \mid \pp, \pq \in \Pset}$:
	\begin{enumerate}
		\item\label{wbisim:1} if $\cC \cto{\mu} \cCp$, then there is a $\nNp$ such that $\nN \pto{\mu} \nNp$ and $(\cCp,\nNp)\in\mathcal{R}$;
		\item\label{wbisim:2} if $\nN \pto{\mu} \nNp$, then there is a $\cCp$ such that $\cC \cto{\mu} \cCp$ and $(\cCp,\nNp)\in\mathcal{R}$.
	\end{enumerate}
	A \net and a \chor are \defn{bisimilar}, written $\cC \wbisim \nN$, if there is a weak bisimulation $\mathcal{R}$ such that $(\cC,\nN)\in\mathcal{R}$.
\end{definition}

\begin{restatable}[EPP]{theorem}{EPP}\label{thm:logicEPP}
	Let $\cC$ be a choreography. If $\DofC\cC$ is a well-defined proof in $\nChorL$ with conclusion $\Gamma$, 
	then there exists a network $\nN$ {\LF} such that $\nN \bisim \cC$ and $\Gamma = \fof{\nN}$.
\end{restatable}
\begin{proof}
	By \Cref{rmk:byconstruction} there is $\nN$ such that $\fof \nN = \Gamma$.
	Since $\DofC\cC$ is a proof of $\fof \nN$ so by \Cref{cor:lf_iif_proof} $\nN$ is \LF.
	To show that $\cC \bisim \nN$ we show that the relation
	$\abisim = \set{(\nN, \cC) \mid \fof{\nN} \text{ is the conclusion of } \DofC{\cC}}$
	is a bisimulation.
	We verify the two bisimulation conditions of \Cref{def:wbisim}.

	\begin{enumerate}

		\item 
		
		We proceed by case analysis on the bottom-most rule $\rrule$ in the derivation of the transition $\cC \cto{\mu} \cCp$ in the operational semantics of choreographies, by letting 
		$\fof{\nN} = \Gamma$ being the conclusion of $\DofC{\cC}$
		and
		$\fof{\nNp} = \Gamma'$ being the conclusion of $\DofC{\cCp}$.

		\begin{itemize}
			\item 
			If $\rrule = \comrule$, then $\cC = \cgencom \seq \cC[1]$ and the bottom-most rule of $\DofC{\cC}$ is an instance of $\Ccomr$ with conclusion $\fof{\nN}=\fof{\pdo[p]{\psend qe \pseq \pP} \ppar \pdo[q]{\precv px \pseq \pQ} \ppar \nM}$
			and premise 
			$\fof{\pdo[p]{\pP} \ppar \pdo[q]{\pQ\psubst{e}{x}} \ppar \nM}=\fof{\nNp}$, the conclusion of $\DofC{\cCp}$.
			We conclude since $\nN\pto{\tocom peqx}\nNp$.

			\item 
			If $\rrule = \choirule$ \resp{$\rrule = \labrule$} we conclude similarly to the previous case, but considering the fact that the bottom-most rule in $\DofC{\cC}$ is a $\Cchor$ \resp{$\Clabr$} and $\nN\pto{\tolab pq}\nNp$ \resp{$\nN\pto{\tochoi p}\nNp$}.

			\item 
			If $\rrule = \tensrulel$ \resp{$\tensruler$}, then
			$\cC = \cC[l] \ctens \cC[r]$ and 
			the bottom-most rule of $\DofC{\cC}$ is $\Cmixr$ with conclusion $\vdash \fof{\nN[l] \ppar \nN[r]}$ and two premises $\vdash \fof{\nN[l]}$ and $\vdash \fof{\nN[r]}$, each conclusion of a proof $\DofC{\cC[l]}$ and $\DofC{\cC[r]}$, respectively.
			By inductive hypothesis on $\DofC{\cC[l]}$ \resp{$\DofC{\cC[r]}$},
			there exists $\nNp[l]$ \resp{$\nNp[r]$} with
			$\nN[l] \pto{\mu} \nNp[l]$ \resp{$\nN[r] \pto{\mu} \nNp[r]$}
			and $\fof{\nNp[l]} = \Gamma'_l$ \resp{$\fof{\nNp[r]} = \Gamma'_r$},
			the conclusion of $\DofC{\cCp[l]}$ \resp{$\DofC{\cCp[r]}$}.
			We conclude because we have 
			$\nN[l] \ppar \nN[r] \pto{\mu} \nNp[l] \ppar \nN[r]$
			\resp{$\nN[l] \ppar \nN[r] \pto{\mu} \nN[l] \ppar \nNp[r]$}
			thanks to the rule $\pnetrule$ in the operational semantics of networks.

			\item If $\rrule = \delcomrule$, then $\cC = \cgencom \seq \cC[1]$
			where $\pp,\pq\notin\pnof{\mu}$, and the bottom-most rule of $\DofC{\cC}$
			is a $\Ccomr$.
			By inductive hypothesis on $\DofC{\cC[1]}$, there exists $\nNp[1]$
			with $\nN[1] \pto{\mu} \nNp[1]$ and $\fof{\nNp[1]} = \Gamma'_1$,
			the conclusion of $\DofC{\cCp[1]}$.
			Since $\pnof{\mu} \cap \set{\pp,\pq} = \emptyset$, we have
			$\nN(\pp) = \nN[1](\pp)$ for all $\pp \in \pnof{\mu}$,
			hence $\nN \pto{\mu} \nNp$.
			Moreover, $\DofC{\cCp}$ is obtained from $\DofC{\cCp[1]}$ by
			appending the same $\Ccomr$ as bottom-most rule, so $\Gamma' =
			\fof{\nNp}$.

			\item If $\rrule = \delchoirule$ we conclude similarly to the previous case, but assuming that the bottom-most rule of $\DofC{\cC}$ is $\Cchor$ followed by instances of $\Clabr$ and that the transition $\nN \pto{\mu} \nNp$ is derived in the network semantics using the rule $\pprocrule$.

			\item If $\rrule = \unfoldrule$, we conclude by inductive hypothesis,
			noting that the corresponding network step uses rule $\pprocrule$.
		\end{itemize}

		\item 
		We proceed by induction on the maximal distance $d$ in the derivation $\DofC{\cC}$ of the rule corresponding to the transition with label $\mu$ from the root of the derivation $\DofC{\cC}$ in all its branches.

		If $d=0$, then the bottom-most rule of $\DofC{\cC}$ directly encodes the transition labeled by $\mu$: if $\mu = \tocom peqx$ \resp{$\tolab pq$,
		$\tochoi p$}, the bottom-most rule is $\Ccomr$ \resp{$\Cchor$ preceded
		by instances of $\Clabr$, $\Clabr$}, and we conclude with a reasoning similar to the one we used in the previous point.

		For the inductive step, $d > 0$: the bottom-most rule of $\DofC{\cC}$
		encodes a transition, whose process names of its label are not in $\pnof{\mu}$; 
		it is one of $\Cmixr$, $\Ccomr$, $\Cchor$, $\Clabr$, or $\Crec$.
		In each subcase, applying the inductive hypothesis to the relevant
		sub-derivation (where the maximum depth of the rule for $\mu$ decreases
		by one) and concluding by the same argument as in the $\tensrulel$,
		$\delcomrule$, $\delchoirule$, and $\unfoldrule$ in the previous point, respectively.

	\end{enumerate}
\end{proof}

\subsection{Extraction, logically}\label{sec:proofsAsChoreos:extraction}

We design an algorithm for extracting a choreography from a proof of a sequent encoding a \LF network, and we show that the extracted choreography is bisimilar to the original network.

\begin{definition}
	Let $\dD$ be a proof in $\nChorL$ of a sequent $\Gamma$.
	The \defn{choreography extracted from $\dD$} is the choreography $\chRCof[\dD]$ defined by induction on the structure of $\dD$ as described in \Cref{fig:DerToChor}.
\end{definition}

\begin{figure}[t]\centering
	\adjustbox{max width=\textwidth}{$\begin{array}{c}
		\chRCof[\vlinf{\Caxr}{}{\vdash 
			\fof{\nnil}
		}{}
		]=\cnil
	\qquad\qquad
		\chRCof[
			\vlderivation{
				\vliin{\Cmixr}{}{
					\vdash \fof{\nN[1]\ppar \nN[2]}
				}{
					\vlpr{\dD_1}{}{\vdash \fof{\nN[1]}}
				}{
					\vlpr{\dD_2}{}{\vdash \fof{\nN[2]}}
				}
			}
		]
		= 
		\chRCof[\vlderivation{
			\vlpr{\dD_1}{}{\vdash \fof{\nN[1]}}
		}] 
		\ctens 
		\chRCof[
			\vlderivation{
			\vlpr{\dD_2}{}{\vdash \fof{\nN[2]}}
		}]
	\\\\
		\chRCof[
			\vlderivation{
				\vlin{\Cwr}{}{
					\vdash \lunit,\fof{ \nN}
				}{
					\vlpr{\dD_1}{}{\vdash \fof{\nN}}
				}
			}
		]=
		\chRCof[
			\vlderivation{
				\vlpr{\dD_1}{}{\vdash \fof{\nN}}
			}
		]
	\qquad
		\chRCof[
			\vlderivation{
				\vlin{\ldefr}{}{
					\vdash 
					\fof{\pdo{\pX} \ppar \nN}
				}{ 
					\vlpr{\dD}{}{
						\vdash 
						\fof{\pdo{\pP_{\pX}} \ppar \nN}
					}
				}
			}
		]
		= 
		\chRCof[
			\vlderivation{
				\vlpr{\dD}{}{
					\vdash 
					\fof{\pdo{\pP[\pX]} \ppar \nN}
				}
			}
		]
	\\\\
		\chRCof[
			\vlderivation{
				\vlin{\Crec}{}{
					\vdash 
					\fof{
						\pdo{\pX[{\pp_1}]} 
						\ppar \cdots\ppar
						\pdo{\pX[{\pp_n}]} 
						\ppar
						\nN
					}
				}{
					\vlin{}{}{
						\vdash 
						\fof{
							\pdo{\pP[{\pp_1}]} 
							\ppar \cdots\ppar
							\pdo{\pP[{\pp_n}]} 
							\ppar
							\nN
						}
					}{
						\vlpr{\dD_1}{}{
							\vdash \Gamma
						}
					}
				}
			}
		]
		= 
		\cX
		\qquad\text{where }
		\procDef{\cX}{\chRCof[
			\vlderivation{
				\vlin{}{}{
					\vdash 
					\fof{
						\pdo{\pP[{\pp_1}]} 
						\ppar \cdots\ppar
						\pdo{\pP[{\pp_n}]} 
						\ppar
						\nN
					}
				}{
					\vlpr{\dD_1}{}{
						\vdash \Gamma
					}
				}
			}
		]}\in\cprocset
	\\\\
		\chRCof[\vlderivation{\vlin{\Ccomr}{}{
				\vdash 
				\fof{
					\pdo{\psend qe \pseq \pP}
					\ppar
					\pdo[q]{\precv px \pseq \pQ}
					\ppar \nN	
				}
			}{
				\vlpr{\dD'\lpsubst ex}{}{
					\vdash 
					\fof{
						\pdo{\pP}
						\ppar
						\pdo[q]{\pQ \psubst ex}
						\ppar 
						\nN
					}
				}
			}
		}]=
		\cgencom \seq \chRCof[\vlderivation{\vlpr{\dD'}{}{
			\vdash 
			\fof{
				\pdo{\pP}
				\ppar
				\pdo[q]{\pQ}
				\ppar 
				\nN
			}
		}}]
	\\\\
		\chRCof[
			\vlderivation{
				\vlin{\Cchor}{}{\vdash 
					\fof{
						\pdo[p]{\matpsendl[\lab\in\aLset]q{\lab: \pP[\lab]}}
						\ppar
						\pdo[q]{\matprecvl[\lab\in\aLset']p{\lab: \pQ[\lab]}}
						\ppar \nN
					}
					}{
					\multiprem{\vlin{\Clabr}{}{
						\vdash 
						\fof{
							\pdo[p]{\matpsendl[]q{\lab_i: \pP_{\lab_i}}}
							\ppar
							\pdo[q]{\matprecvl[\lab\in\aLset']p{\lab: \pQ[\lab]}}
							\ppar \nN
						}
					}{
						\vlpr{\dD_{\lab}}{}{
							\vdash 
							\fof{
								\pdo{\pP[\lab_i]}
								\ppar
								\pdo{\pQ[\lab_i]}
								\ppar 
								\nN
							}
						}	
					}}{\lab\in L}
				}	
			}
		]
		=
		\cchoi{p}{\aLset}{q}{
			\lab : {\chRCof[\vlderivation{
				\vlpr{\dD_{\lab}}{}{
					\vdash 
					\Gamma,
					\fof{\pdo{\pP[{\lab}]}}, \fof{\pdo[q]{\pQ[{\lab}]}}
			}}]}
			\mid \lab \in \aLset
		}
	\end{array}$}

	\caption{
		Translation from $\nChorL$ proofs to choreographies.
	}
	\label{fig:DerToChor}
\end{figure}

\begin{lemma}\label{lem:bisimcoreo}
	The proof equivalence $\peq$ induces a bisimulation of choreographies.
	That is, if two derivations $\dD$ and $\dD'$ in $\nChorL$ are such that $\dD \peq \dD'$, 
	then $\chRCof[\dD] \bisim \chRCof[\dD']$.
\end{lemma}
\begin{proof}
	We observe that the permutations in \Cref{fig:equivTCHORL} preserve \niceness. First, $\Cmixr$ and $\Crec$ cannot 
	be permuted with each other. Second, only finitely many steps are permuted; therefore, an infinite derivation still 
	contains infinitely many progress points.
	If $\dD \peq \dD'$, then there is a finite sequence of rule permutations as in \Cref{fig:equivTCHORL} transforming $\dD$ into $\dD'$.
	We conclude by induction on the number of rule permutations in the sequence, 
	after observing that each rule permutation preserves the bisimilarity of the extracted choreographies.
\end{proof}


\begin{restatable}[Extraction]{theorem}{EXTR}\label{thm:extr}
	Let 
	$\dD\in\nChorL$ 
	and let 
	$\nN_\dD$ be the network such that $\fof{\nN_\dD}$ is the conclusion of $\dD$.
	Then $\chRCof[\dD] \bisim \nN$.
\end{restatable}
	
\begin{proof}
	We define the following relation between choreographies and networks, 
	$$
	\mathcal{R} 
	= 
	\set{
		(\cC, \nN) \mid 
		\text{ there is $\dD \in \nChorL$ with conclusion $\fof{\nM}$ such that 
		$\cC \bisim \chRCof[\dD]$, and $\nM \bisim \nN$}
	}$$
	and we show that it is a bisimulation (see \Cref{def:wbisim}) containing $(\chRCof[\dD], \nN)$.

	\begin{enumerate}
	\item\label{EXTR:CtoN}
		Proof that if $\cC \cto{\mu} \cCp$, then there is a $\nNp$ such that $\nN \pto{\mu} \nNp$ and $(\cCp,\nNp)\in\mathcal{R}$:
		
		If $\cC \cto{\mu} \cCp$, then there is $\cB$ such that 
		$\cC \bisim \chRCof[\dD] \cto{\mu} \cB$ and $\cCp \bisim \cB$. 
		We can then construct a proof $\dD''\in\nChorL$ with conclusion $\fof{\nNpp}$
		such that $\cB \bisim \chRCof[\dD'']$ and, ${\nN} \pto{\mu} \nNp \bisim {\nNpp}$.
		By definition of $\mathcal{R}$, $(\cB,\nNpp)\in\mathcal{R}$, and we can conclude using transitivity of bisimilarity that $(\cCp,\nNp)\in\mathcal{R}$.

		We proceed by case analysis on the label $\mu$.
		Thanks to \Cref{lem:bisimcoreo} and guardedness of \chors, we can always assume that the bottom of $\dD$ is 
		made of possibly some instances of $\Cmixr$ and $\Crec$ only involving the process names in $\mu$, followed by either $\Ccomr$, $\Cchor$ or $\Clabr$, accordingly wether $\mu$ is $\tocom peqx$, $\tolab pq$ or $\tochoi p$.

		If $\mu = \tocom peqx$, then
		$$
		\dD=
		\vlderivation{
			\vliid{\Cmixr}{}{
				\vdash \fof{\nM} 
			}{
				\vlid{\unfoldr}{}{
					\vdash \fof{\nM[1] \ppar \nM[2]} 
				}{
					\vlin{\Ccomr}{}{
						\vdash \fof{\pdo{\psend qe \pseq \pP} \ppar \pdo[q]{\precv qx \pseq \pQ }\ppar \nM[2]}
					}{
						\vlpr{\dD_{\mathsf{L}}}{}{\vdash \fof{\pdo{\pP} \ppar \pdo[q]{\pQ\psubst ex} \ppar \nM[2]}}
					}
				}
			}{\vlpr{\dD_{\mathsf{R}}}{}{\fof{\nM[3]}}}
		}
		$$
		where the $\ldefr$-rules may only involve formulas where $\pp$ or $\pq$ occur, and so
		$$\nM_1\in\set{
			\pdo[p]{\pX} \ppar \pdo[q]{\pY}  
		\qomma
			\pdo[p]{\pX} \ppar \pdo[q]{\precv qx \pseq \pQ}  
		\qomma
			\pdo[p]{\psend qe \pseq \pP} \ppar \pdo[q]{\pY} 
		\qomma
			\pdo[p]{\psend qe \pseq \pP} \ppar \pdo[q]{\precv qx \pseq \pQ}
		}$$
		with 
		$\pX \triangleq \psend qe \pseq \pP $ and 
		$\pY \triangleq \precv qx \pseq \pQ $.
		We can then let $\dD''$ be defined as follows:
		$$\dD'' = 
		\vlderivation{
			\vliid{\Cmixr}{}{
				\vdash \fof{\nNpp} 
			}{
				\vlpr{\dD_{\mathsf{L}}}{}{\vdash \fof{\pdo{\pP} \ppar \pdo[q]{\pQ\psubst ex} \ppar \nM[2]}}
			}{\vlpr{\dD_{\mathsf{R}}}{}{\fof{\nM[3]}}}
		}
		$$
		with $\nNpp=\left(\pdo[p]{\pP} \ppar \pdo[q]{\pQ\psubst ex} \ppar \nM[2] \ppar \nM[3]\right)$.
		We observe that
		$\chRCof[\dD]
		\cto{\tocom peqx} 
		\chRCof[\dD'']
		=
		\cB
		$
		and  
		we conclude since 
		$\nN \bisim \nM \pto{\tocom peqx} \nNpp$.

		The other cases are similar, with case $\tolab pq$ needing to prove (by definition of $\chRCof[\dD]$, and the definition of the rule $\delchoirule$) that 
		if \ $\chRCof[\dD] = \cgenchois $, then
		$\cgenchoiEXT \bisim  \chorpoly{\choi{p}{\set{\lab_i}}{q}{\lab_i : \cC[\lab_i]}}$ for all $\lab_i \in \aLset$.
		
	\item 
		Proof that if $\nN \pto{\mu} \nNp$, then there is a $\cCp$ such that $\cC \cto{\mu} \cCp$ and $(\cCp,\nNp)\in\mathcal{R}$:

		If $\nN \pto{\mu} \nNp$ and $\nN \bisim \nM$ with $\fof\nM$ being the conclusion of $\dD$, then, we can again assume that the bottom-most rule in $\dD$ is either $\Ccomr$, $\Cchor$ or $\Clabr$ (possibly preceded by instances of $\Cmixr$ and $\Crec$), accordingly whether $\mu$ is $\tocom peqx$, $\tolab pq$ or $\tochoi p$
		and that $\nM \pto{\mu} \nMp$ with $\nMp$ being the child of $\nM$ corresponding to the transition labeled by $\mu$.

		To conclude, it suffices to show that the subderivation $\dD'$ with conclusion $\fof{\nMp}$ is such that $\chRCof[\dD] \cto{\mu} \cCp \bisim \chRCof[\dD']$.
		The argument is similar to the one we used in Item~(\ref{EXTR:CtoN}).
		\qedhere
	\end{enumerate}
\end{proof}

\begin{theorem}[Completeness]\label{thm:extrbisim}
	Let $\nN$ be a \net. Then 
	$$\nN \mbox{ is \LF} \Rightarrow  \mbox{ there is a choreography $\cC$ such that } \pnof{\cC} = \pnof{\nN} \mbox{ and } \cC \wbisim \nN.$$	
\end{theorem}
\begin{proof}
	If $\nN$ is \LF, then by \Cref{thm:LFinTCHORL}, $\fof \nN$ admits a proof $\dD\in\nChorL$. We conclude by \Cref{thm:extr} by letting $\cC =\chRCof[\dD]$.
\end{proof}

\begin{remark}\label{rmk:nice}
	The \defn{\nice} condition on proofs ensures that, when extracting a choreography from a proof, the extraction procedure immediately separate the independent components of a choreography using the $\Cmixr$-rule.
	This guarantees that independent components give rise to disjoint choreographies composed by $\ctens$, as in \Cref{ex:tensor}.
	For example, while the network 
	$$
	\pdo[p]{\psend qe \seq \pX[\pp]}
	\ppar
	\pdo[q]{\precv px \seq \pY[\pq]}
	\ppar
	\pdo[r]{\psend rf \seq \pX[\pr]}
	\ppar
	\pdo[s]{\precv sy \seq \pY[\ps]}
	$$
	is bisimilar to both the following choreographies:
	$$
	\begin{array}{l@{\qquad}l}
		\cC[1]= \ccom pe qx \seq \ccom rfsy \seq \cX[1]
	&
		\text{with } \procDef{\cX[1]}{\ccom pe qx \seq \ccom rfsy \seq \cX[1]}
	\\
		\cC[\ltens]= (\ccom pe qx \seq\cX[l])\ctens (\ccom rfsy \seq \cX[r])
	&
	\text{with } 
		\procDef{\cX[l]}{\ccom pe qx \seq \cX[l]} 
	\text{ and }
		\procDef{\cX[r]}{\ccom rfsy \seq \cX[r]}
	\end{array}
	$$
	the extraction procedure will always produce the second choreography, which is the one that better reflects disjoint the structure of the network.
\end{remark}

\section{Conclusions and Future Work}\label{sec:conc}

In this paper we provided a new framework for the study of choreographies and their relation with process calculi, based on the proof theory of linear logic.
We showed that projection and extraction can be studied in a uniform way, by establishing a correspondence between specific proofs in the sequent calculus and trees of executions in the network calculus.

\looseness=-1
The main insight, is that adopting a logic-programming oriented view of process calculi, allows studying the properties of processes by reasoning on the corresponding proofs in the logic, and vice versa, in a quite natural way.
This characterization is simple and elegant.
Lock-freedom, which is a fundamental property of networks, can then be characterized in simple terms of existence of proofs in the logic, which, in turn, can be seen as successful execution of a logic program.
The same framework allowed us to define a new endpoint projection with two notable advantages over existing approaches.
First, it dispenses from the merge operator~\cite{montesi:book,DYBH12,BCDLDY08}, a syntactic mechanism traditionally used to reconstruct the knowledge of choice during endpoint projection, whose correctness relies on technically involved metatheoretical arguments~\cite{SY19}.
Secondly, removing merge does not restrict the class of projectable choreographies. 
On the contrary, our proof-theoretic formulation naturally extends it, allowing the projection of choreographies with recursive choices such as the one in \eqref{ex:intro:unprojectable:choreo}, which, to the best of our knowledge, have not been handled by previous projection procedures.

\paragraph*{Perspectives}\looseness=-1
The work presented in this paper provides an effective framework for choreographic extraction.
By implementing proof search algorithms for the logic, it is possible to automatically extract choreographies from networks.
For \LF networks with only tail-recursive calls, the extraction algorithm is guaranteed to terminate, and, based on the complexity results for proof search in $\MALL$ \cite{lincoln:complexity}, it should be possible to implement an extraction algorithm with polynomial-space complexity.%
\footnote{
	In general, the complexity of proof search in $\mu\MALL$ is more complex \cite{das:de:sau:muMALL}, however, the fragment of $\mu\MALL$ captured by $\iPIL$ is quite restricted (no least fixpoints, tail-recursive networks, race-freedom).
}
In fact, the first-order component of the logic is only used in a very limited way in the proof search for sequents encoding networks, and the tail-recursive restriction on networks ensures that only a finite number of formulas (modulo first-order substitutions) can appear in the proof search, thus allowing to implement a polynomial-space proof search algorithm cyclic proofs.

Our also suggests a new generation of choreographic compilers based directly on the proof-theoretic endpoint projection developed in this paper that supports a strictly larger class of choreographies w.r.t., alternative approaches for comparable languages without incurring in runtime costs of auxiliary comminations \cite{PGGLM16,PGSN22} or complex procedures like the merge operator \cite{montesi:book,CMP23} and amendment \cite{CM23}. 

\bibliographystyle{ACM-Reference-Format}
\bibliography{biblio}

\clearpage

\appendix

\section{Details of Proofs}\label{sec:appendix}

\EPP*
\begin{proof}
	By \Cref{rmk:byconstruction} there is $\nN$ such that $\fof \nN = \Gamma$.
	Since $\DofC\cC$ is a proof of $\fof \nN$ so by \Cref{cor:lf_iif_proof} $\nN$ is \LF.
	To show that $\cC \bisim \nN$ we show that the relation
	$\abisim = \set{(\nN, \cC) \mid \fof{\nN} \text{ is the conclusion of } \DofC{\cC}}$
	is a bisimulation.
	We verify the two bisimulation conditions of \Cref{def:wbisim}.

	\begin{enumerate}

		\item 
		
		We proceed by case analysis on the bottom-most rule $\rrule$ in the derivation of the transition $\cC \cto{\mu} \cCp$ in the operational semantics of choreographies, by letting 
		$\fof{\nN} = \Gamma$ being the conclusion of $\DofC{\cC}$
		and
		$\fof{\nNp} = \Gamma'$ being the conclusion of $\DofC{\cCp}$.

		\begin{itemize}
			\item 
			If $\rrule = \comrule$, then $\cC = \cgencom \seq \cC[1]$ and the bottom-most rule of $\DofC{\cC}$ is an instance of $\Ccomr$ with conclusion $\fof{\nN}=\fof{\pdo[p]{\psend qe \pseq \pP} \ppar \pdo[q]{\precv px \pseq \pQ} \ppar \nM}$
			and premise 
			$\fof{\pdo[p]{\pP} \ppar \pdo[q]{\pQ\psubst{e}{x}} \ppar \nM}=\fof{\nNp}$, the conclusion of $\DofC{\cCp}$.
			We conclude since $\nN\pto{\tocom peqx}\nNp$.

			\item 
			If $\rrule = \choirule$ \resp{$\rrule = \labrule$} we conclude similarly to the previous case, but considering the fact that the bottom-most rule in $\DofC{\cC}$ is a $\Cchor$ \resp{$\Clabr$} and $\nN\pto{\tolab pq}\nNp$ \resp{$\nN\pto{\tochoi p}\nNp$}.

			\item 
			If $\rrule = \tensrulel$ \resp{$\tensruler$}, then
			$\cC = \cC[l] \ctens \cC[r]$ and 
			the bottom-most rule of $\DofC{\cC}$ is $\Cmixr$ with conclusion $\vdash \fof{\nN[l] \ppar \nN[r]}$ and two premises $\vdash \fof{\nN[l]}$ and $\vdash \fof{\nN[r]}$, each conclusion of a proof $\DofC{\cC[l]}$ and $\DofC{\cC[r]}$, respectively.
			By inductive hypothesis on $\DofC{\cC[l]}$ \resp{$\DofC{\cC[r]}$},
			there exists $\nNp[l]$ \resp{$\nNp[r]$} with
			$\nN[l] \pto{\mu} \nNp[l]$ \resp{$\nN[r] \pto{\mu} \nNp[r]$}
			and $\fof{\nNp[l]} = \Gamma'_l$ \resp{$\fof{\nNp[r]} = \Gamma'_r$},
			the conclusion of $\DofC{\cCp[l]}$ \resp{$\DofC{\cCp[r]}$}.
			We conclude because we have 
			$\nN[l] \ppar \nN[r] \pto{\mu} \nNp[l] \ppar \nN[r]$
			\resp{$\nN[l] \ppar \nN[r] \pto{\mu} \nN[l] \ppar \nNp[r]$}
			thanks to the rule $\pnetrule$ in the operational semantics of networks.

			\item If $\rrule = \delcomrule$, then $\cC = \cgencom \seq \cC[1]$
			where $\pp,\pq\notin\pnof{\mu}$, and the bottom-most rule of $\DofC{\cC}$
			is a $\Ccomr$.
			By inductive hypothesis on $\DofC{\cC[1]}$, there exists $\nNp[1]$
			with $\nN[1] \pto{\mu} \nNp[1]$ and $\fof{\nNp[1]} = \Gamma'_1$,
			the conclusion of $\DofC{\cCp[1]}$.
			Since $\pnof{\mu} \cap \set{\pp,\pq} = \emptyset$, we have
			$\nN(\pp) = \nN[1](\pp)$ for all $\pp \in \pnof{\mu}$,
			hence $\nN \pto{\mu} \nNp$.
			Moreover, $\DofC{\cCp}$ is obtained from $\DofC{\cCp[1]}$ by
			appending the same $\Ccomr$ as bottom-most rule, so $\Gamma' =
			\fof{\nNp}$.

			\item If $\rrule = \delchoirule$ we conclude similarly to the previous case, but assuming that the bottom-most rule of $\DofC{\cC}$ is $\Cchor$ followed by instances of $\Clabr$ and that the transition $\nN \pto{\mu} \nNp$ is derived in the network semantics using the rule $\pprocrule$.

			\item If $\rrule = \unfoldrule$, we conclude by inductive hypothesis,
			noting that the corresponding network step uses rule $\pprocrule$.
		\end{itemize}

		\item 
		We proceed by induction on the maximal distance $d$ in the derivation $\DofC{\cC}$ of the rule corresponding to the transition with label $\mu$ from the root of the derivation $\DofC{\cC}$ in all its branches.

		If $d=0$, then the bottom-most rule of $\DofC{\cC}$ directly encodes the transition labeled by $\mu$: if $\mu = \tocom peqx$ \resp{$\tolab pq$,
		$\tochoi p$}, the bottom-most rule is $\Ccomr$ \resp{$\Cchor$ preceded
		by instances of $\Clabr$, $\Clabr$}, and we conclude with a reasoning similar to the one we used in the previous point.

		For the inductive step, $d > 0$: the bottom-most rule of $\DofC{\cC}$
		encodes a transition, whose process names of its label are not in $\pnof{\mu}$; 
		it is one of $\Cmixr$, $\Ccomr$, $\Cchor$, $\Clabr$, or $\Crec$.
		In each subcase, applying the inductive hypothesis to the relevant
		sub-derivation (where the maximum depth of the rule for $\mu$ decreases
		by one) and concluding by the same argument as in the $\tensrulel$,
		$\delcomrule$, $\delchoirule$, and $\unfoldrule$ in the previous point, respectively.

	\end{enumerate}
\end{proof}

\EXTR*
\begin{proof}
	We define the following relation between choreographies and networks, 
	$$
	\mathcal{R} 
	= 
	\set{
		(\cC, \nN) \mid 
		\text{ there is $\dD \in \nChorL$ with conclusion $\fof{\nM}$ such that 
		$\cC \bisim \chRCof[\dD]$, and $\nM \bisim \nN$}
	}$$
	and we show that it is a bisimulation (see \Cref{def:wbisim}) containing $(\chRCof[\dD], \nN)$.

	\begin{enumerate}
		\item\label{EXTR:CtoN}
		Proof that if $\cC \cto{\mu} \cCp$, then there is a $\nNp$ such that $\nN \pto{\mu} \nNp$ and $(\cCp,\nNp)\in\mathcal{R}$:
		
		If $\cC \cto{\mu} \cCp$, then there is $\cB$ such that 
		$\cC \bisim \chRCof[\dD] \cto{\mu} \cB$ and $\cCp \bisim \cB$. 
		We can then construct a proof $\dD''\in\nChorL$ with conclusion $\fof{\nNpp}$
		such that $\cB \bisim \chRCof[\dD'']$ and, ${\nN} \pto{\mu} \nNp \bisim {\nNpp}$.
		By definition of $\mathcal{R}$, $(\cB,\nNpp)\in\mathcal{R}$, and we can conclude using transitivity of bisimilarity that $(\cCp,\nNp)\in\mathcal{R}$.

		We proceed by case analysis on the label $\mu$.
		Thanks to \Cref{lem:bisimcoreo} and guardness of \chors, we can always assume that the bottom of $\dD$ is 
		made of possibly some instances of $\Cmixr$ and $\Crec$ only involing the process names in $\mu$, followed by either $\Ccomr$, $\Cchor$ or $\Clabr$, accordingly wether $\mu$ is $\tocom peqx$, $\tolab pq$ or $\tochoi p$.
		\begin{itemize}
			\item
			If $\mu = \tocom peqx$, then
			$$
			\dD=
			\vlderivation{
				\vliid{\Cmixr}{}{
					\vdash \fof{\nM} 
				}{
					\vlid{\unfoldr}{}{
						\vdash \fof{\nM[1] \ppar \nM[2]} 
					}{
						\vlin{\Ccomr}{}{
							\vdash \fof{\pdo{\psend qe \pseq \pP} \ppar \pdo[q]{\precv qx \pseq \pQ }\ppar \nM[2]}
						}{
							\vlpr{\dD_{\mathsf{L}}}{}{\vdash \fof{\pdo{\pP} \ppar \pdo[q]{\pQ\psubst ex} \ppar \nM[2]}}
						}
					}
				}{\vlpr{\dD_{\mathsf{R}}}{}{\fof{\nM[3]}}}
			}
			$$
			where the $\ldefr$-rules may only involve formulas where $\pp$ or $\pq$ occur, and so
			$$\nM_1\in\set{
				\pdo[p]{\pX} \ppar \pdo[q]{\pY}  
			\qomma
				\pdo[p]{\pX} \ppar \pdo[q]{\precv qx \pseq \pQ}  
			\qomma
				\pdo[p]{\psend qe \pseq \pP} \ppar \pdo[q]{\pY} 
			\qomma
				\pdo[p]{\psend qe \pseq \pP} \ppar \pdo[q]{\precv qx \pseq \pQ}
			}$$
			with 
			$\pX \triangleq \psend qe \pseq \pP $ and 
			$\pY \triangleq \precv qx \pseq \pQ $.
			We can then let $\dD''$ be defined as follows:
			$$\dD'' = 
			\vlderivation{
				\vliid{\Cmixr}{}{
					\vdash \fof{\nNpp} 
				}{
					\vlpr{\dD_{\mathsf{L}}}{}{\vdash \fof{\pdo{\pP} \ppar \pdo[q]{\pQ\psubst ex} \ppar \nM[2]}}
				}{\vlpr{\dD_{\mathsf{R}}}{}{\fof{\nM[3]}}}
			}
			$$
			with $\nNpp=\left(\pdo[p]{\pP} \ppar \pdo[q]{\pQ\psubst ex} \ppar \nM[2] \ppar \nM[3]\right)$.
			We observe that
			$\chRCof[\dD]
			\cto{\tocom peqx} 
			\chRCof[\dD'']
			=
			\cB
			$
			and  
			we conclude since 
			$\nN \bisim \nM \pto{\tocom peqx} \nNpp$.
			

			\item If $\mu = \tochoi p$, then we let 
			$$
			\dD =
			\vlderivation{
				\vliid{\Cmixr}{}{
					\fof{\nM}
				}{
					\vlid{\unfoldr}{}{
						\vdash \fof{\nM[1] \ppar \nM[2]}
					}{
						\vlin{\Cchor}{}{
							\vdash \fof{\pdo{\pgensendl} \ppar \pdo[q]{\matprecvl[\lab\in\aLset']p{\lab: \pQ[\lab]}} \ppar \nM[2]}
						}{
							\multiprem{
								\vlin{\Clabr}{}{
									\vdash \fof{\pdo{\pgenselectl} \ppar \pdo[q]{\matprecvl[\lab\in\aLset']p{\lab: \pQ[\lab]}} \ppar  \nM[2]}
										}{
										\vlpr{\dD_{\lab}}{}{
											\vdash \fof{\pdo{\pP[\lab_i]} \ppar \pdo{\pQ[\lab_i]} \ppar  \nM[2]}
										}
									}
							}{\lab_i \in L}
						}
					}
				}{\vlpr{\dD_{\mathsf{R}}}{}{\fof{\nM[3]}}}
			}
			$$
			where the $\ldefr$-rules may only involve formulas where $\pp$ or $\pq$ occur, and so
			$$
			\adjustbox{max width=\textwidth}{$\nM[1]\in\set{\pdo[p]{\pX} \ppar \pdo[q]{\pY}\;,\;
					\pdo[p]{\pgenselectl}  \ppar \pdo[q]{\pY}\;,\;
					\pdo[p]{\pX} \ppar \pdo[q]{\matprecvl[\lab\in\aLset']p{\lab: \pQ[\lab]}} \;,\;
					\pdo[p]{\pgenselectl}  \ppar \pdo[q]{\matprecvl[\lab\in\aLset']p{\lab: \pQ[\lab]}}}$}$$ 
			with $\pX \triangleq \pgensendl$ and $ \pY \triangleq {\matprecvl[\lab\in\aLset']p{\lab: \pQ[\lab]}}$.
			We can then let $\dD''$ be defined as follows:
			$$\dD'' =
			\vlderivation{
				\vliid{\Cmixr}{}{
					\fof{\nNpp}
				}{
					{
						\vlin{\Cchor}{}{
							\vdash \fof{\pdo{\pgensendl} \ppar \pdo[q]{\matprecvl[\lab\in\aLset']p{\lab: \pQ[\lab]}} \ppar \nM[2]}
						}{
							\vlin{\Clabr}{}{
								\vdash \fof{\pdo{\pgenselectl} \ppar \pdo[q]{\matprecvl[\lab\in\aLset']p{\lab: \pQ[\lab]}} \ppar  \nM[2]}
							}{
								\vlpr{\dD_{\lab_i}}{}{
									\vdash \fof{\pdo{\pP[\lab_i]} \ppar \pdo{\pQ[\lab_i]} \ppar  \nM[2]}
								}
							}
						}
					}
				}{\vlpr{\dD_{\mathsf{R}}}{}{\fof{\nM[3]}}}
			}
			$$
			with 
			$\nNpp=\pdo{\pgensendl} \ppar \pdo[q]{\matprecvl[\lab\in\aLset']p{\lab: \pQ[\lab]}} \ppar  \nM[2] \ppar \nM[3]$.
			We observe that, by \Cref{lem:choibisim}, $\chRCof[\dD]  \cto{\tochoi p} \cB \bisim \chRCof[\dD'']$ and we conclude since $\nN \bisim \nM \pto{\tochoi p} \nNpp$.
			\item If $\mu = \tolab pq$, then 
			we let
			$$\adjustbox{max width=\textwidth}{$
				\dD = 
				\vlderivation{
					\vliid{}{}{
						\vdash \fof \nM
					}{
						\vlid{}{}{
							\vdash \fof{\nM[1] \ppar  \nM[2]} 
						}{
							\vlin{\Cchor}{}{
								\vdash \fof{\pdo{\pgensendl} \ppar \pdo[q]{\matprecvl[\lab\in\aLset']p{\lab: \pQ[\lab]}} \ppar \nM[2]}
							}{	
								\vlin{\Clabr}{}{
									\vdash \fof{\pdo{\pgenselectl} \ppar \pdo[q]{\matprecvl[\lab\in\aLset']p{\lab: \pQ[\lab]}} \ppar  \nM[2]}
								}{
									\vlpr{\dD_{\lab_i}}{}{
										\vdash \fof{\pdo{\pP[\lab_i]} \ppar \pdo{\pQ[\lab_i]} \ppar  \nM[2]}
									}
								}
							}
						}
					}{\vlpr{\dD_{\mathsf{R}}}{}{\fof{\nM[3]}}}
				}
			$}$$ 
			where $\aLset = \set{\lab_i}$ and the $\ldefr$-rules may only involve formulas where $\pp$ or $\pq$ occur, and so
			$$\adjustbox{max width=\textwidth}{$
				\nM[1]
				\in
				\set{
					\pdo[p]{\pX} \ppar \pdo[q]{\pY}
				\;,\;
					\pdo[p]{\pgenselectl} \ppar \pdo[q]{\pY}
				\;,\;
					\pdo[p]{\pX} \ppar \pdo[q]{\matprecvl[\lab\in\aLset']p{\lab: \pQ[\lab]}}
				\;,\;
					\pdo[p]{\pgenselectl} \ppar \pdo[q]{\matprecvl[\lab\in\aLset']p{\lab: \pQ[\lab]}}
				}
			$}$$ 
			with $\pX \triangleq \pgensendl$ and $ \pY \triangleq {\matprecvl[\lab\in\aLset']p{\lab: \pQ[\lab]}}$.
			We can then let $\dD''$ be defined as follows:
			$$\adjustbox{max width=\textwidth}{$
			\dD'' = 
			\vlderivation{
				\vliid{}{}{
					\vdash \fof \nNpp
				}{
					\vlpr{\dD_{\lab_i}}{}{
						\vdash \fof{\pdo{\pP[\lab_i]} \ppar \pdo{\pQ[\lab_i]} \ppar  \nM[2]}
					}
				}{\vlpr{\dD_{\mathsf{R}}}{}{\fof{\nM[3]}}}
			}
			$}$$ 
			with $\nNpp = \fof{\pdo{\pP[\lab_i]} \ppar \pdo{\pQ[\lab_i]} \ppar  \nM[2] \ppar \nM[3]}$.
			We observe that
			$\chRCof[\dD]
			\cto{ \tolab pq} 
			\chRCof[\dD'']
			=
			\cB
			$
			and  
			we conclude since 
			$\nN \bisim \nM \pto{ \tolab pq} \nNpp$.
		\end{itemize}

		\item 
		Proof that if $\nN \pto{\mu} \nNp$, then there is a $\cCp$ such that $\cC \cto{\mu} \cCp$ and $(\cCp,\nNp)\in\mathcal{R}$:

		If $\nN \pto{\mu} \nNp$ and $\nN \bisim \nM$ with $\fof\nM$ being the conclusion of $\dD$, then, wa can again assume that the bottom-most rule in $\dD$ is either $\Ccomr$, $\Cchor$ or $\Clabr$ (possibly preceded by instances of $\Cmixr$ and $\Crec$), accordingly wether $\mu$ is $\tocom peqx$, $\tolab pq$ or $\tochoi p$
		and that $\nM \pto{\mu} \nMp$ with $\nMp$ being the child of $\nM$ corresponding to the transition labeled by $\mu$.

		To conclude, it suffices to show that the subderivation $\dD'$ with conclusion $\fof{\nMp}$ is such that $\chRCof[\dD] \cto{\mu} \cCp \bisim \chRCof[\dD']$.

		We conclude by an argument similar to the one we used in Item~(\ref{EXTR:CtoN}) above.
	\end{enumerate}
\end{proof}

\end{document}